\documentclass[10pt]{article}
\usepackage{authblk}
\usepackage{graphicx}
\usepackage{subcaption}
\usepackage{physics}
\usepackage{amsmath}
\usepackage{amssymb}
\usepackage{amsthm}
\usepackage[
    colorlinks=true,
    citecolor=blue,
    linkcolor=black,
    urlcolor=black
]{hyperref}
\usepackage{cleveref}
\usepackage{enumitem}
\usepackage{algorithm}
\usepackage{algpseudocode}
\usepackage{caption}
\usepackage[table]{xcolor}

\newcommand{\copymajorsum}{\mathbin{\boxplus}}

\algrenewcommand\algorithmicrequire{\textbf{Input:}}
\algrenewcommand\algorithmicensure{\textbf{Output:}}

\newcommand{\algstep}[2]{%
    \State \parbox[t]{\dimexpr\linewidth-\algorithmicindent\relax}{%
        \hangindent=2em
        \hangafter=1
        \textbf{#1} #2%
    }%
}

\newcommand{\algstepind}[2]{%
    \State \parbox[t]{\dimexpr\linewidth-\algorithmicindent\relax}{%
        \hangindent=4em
        \hangafter=1
        \quad \textbf{#1} #2%
    }%
}

\newtheorem{theorem}{Theorem}
\newtheorem{definition}{Definition}
\newtheorem{lemma}{Lemma}

\newtheorem{remark}{Remark}
\newtheorem{example}{Example}

\crefname{section}{Section}{Sections}
\crefname{appendix}{Appendix}{Appendices}
\crefname{equation}{Eqn.}{Eqns.}
\crefname{figure}{Figure}{Figures}
\crefname{theorem}{Theorem}{Theorems}
\crefname{lemma}{Lemma}{Lemmas}
\crefname{definition}{Definition}{Definitions}
\crefname{corollary}{Corollary}{Corollaries}
\crefname{remark}{Remark}{Remarks}
\crefname{example}{Example}{Examples}
\crefname{algorithm}{Algorithm}{Algorithms}

\usepackage[
    backend=biber,
    style=numeric-comp,
    sorting=none,
    doi=false,
    isbn=false,
    url=false
]{biblatex}

\newcommand{\bluecitebrackets}[1]{%
    \textcolor{blue}{\mkbibbrackets{#1}}%
}

\DeclareCiteCommand{\cite}[\bluecitebrackets]
    {\usebibmacro{cite:init}%
    \usebibmacro{prenote}}
    {\usebibmacro{citeindex}%
    \usebibmacro{cite:comp}}
    {}
    {\usebibmacro{cite:dump}%
    \usebibmacro{postnote}}

\usepackage{geometry}
\title{Distillation of \(N\)-Qubit Stabilizer States on a Star Network Topology}

\author[1]{Theodore M. Mahaffey}
\author[2]{Chaohan Cui}
\author[2]{Saikat Guha}
\author[3]{Murphy Yuezhen Niu}

\affil[1]{Department of Physics, University of California, Santa Barbara, CA 93106, USA}
\affil[2]{Department of Electrical and Computer Engineering, University of Maryland, College Park, MD 20742, USA}
\affil[3]{Department of Computer Science, University of California, Santa Barbara, CA 93106, USA}

\date{}

\renewenvironment{abstract}
    {\small}
    {}

\begin{document}

\maketitle

\begin{center}
\begin{minipage}{0.85\textwidth}
\begin{abstract}

We introduce an entanglement distillation protocol that utilizes an arbitrary $[[n,k,d]]$ stabilizer code to convert $n$ raw copies of an $N$-qubit Greenberger-Horne-Zeilinger (GHZ) state into $k$ logical copies in the presence of Pauli noise. Our explicit formulation of the scheme on a star network topology is efficiently scalable to arbitrary $N$, requires only local operations, and minimizes the impact of errors on idle physical qubits, enabling the distribution of high-fidelity entangled logical states between any number of quantum processors. We report the results of numerical experiments using the 5-qubit code, toric code, and $[[144,12,12]]$ bivariate bicyle code to protect against independent single-qubit depolarizing noise on all qubits. Our encoding scheme improves the final GHZ state fidelity relative to a single bare GHZ state sent over the same noisy channel for depolarizing rates up to or above $6 \%$ for all codes tested. Additionally, we use the stabilizer formalism to prove that our protocol can be applied to any $N$-qubit Calderbank-Shor-Steane (CSS) stabilizer state, from which the GHZ state emerges as a special case. This result explicitly establishes the working principle of previous stabilizer-based Bell pair and GHZ state distillation schemes, and generalizes such protocols beyond the traditionally considered states. We also show that some stabilizer codes can distill any stabilizer state with our protocol.

\end{abstract}
\end{minipage}
\end{center}
\vspace{0.5em}

\section{Introduction}
\label{sec:intro}

Distributing high-fidelity, multipartite entangled states across spatially separated processors is a crucial task that a modular quantum computing architecture must achieve. To combat the challenge posed by inevitable noise in Bell state preparation, entanglement purification and entanglement distillation have been studied in the two-qubit setting for decades \cite{aschauer2004_Quantum, bennett1996_Mixed, maneva2000_Improved, dur2007_Entanglement}. In general, such schemes convert $n$ noisy Bell pairs into $k < n$ higher fidelity Bell pairs. While they are useful for many applications, Bell pairs are limited by the basic fact that they involve only two qubits. To perform non-trivial computations and deploy distributed quantum error correction \cite{clayton2025_Distributed} on a modular quantum computer, it is necessary to distribute entangled states across many ($N > 2$) nodes. For example, in a realization of the surface code on a distributed architecture where the physical qubits of one code block lie in many spatially separated processors, weight 4 stabilizer measurements can be performed by a fault-tolerant syndrome extraction circuit that uses an ancillary 4-qubit Greenberger-Horne-Zeilinger (GHZ) state as a resource state \cite{singh2025_Modular}.

There exist schemes that build an $N$-qubit GHZ state using Bell pairs as the initial resource state, but they require $N-1$ consecutive successful Bell state measurements \cite{bartolucci2021_Creation}. The success probability of the standard method for an individual Bell state measurement is upper bounded by $50$\%, \cite{calsamiglia2001_Maximum}, while boosted Bell state measurement schemes can achieve success rates of $75\%$ \cite{ewert2014_3,hauser2025_Boosted}. In either case, the efficiency of $N$-qubit GHZ state preparation with this scheme is exponentially limited by the non-unit efficiency of Bell state measurements, which in turn limits the clock rate of a quantum network or modular quantum computer relying on distributed GHZ states. Motivated by this, and by experimental demonstrations of the preparation and distribution of GHZ states \cite{li2020_Multipartite,luo2015_Generations,zhao2004_Experimental,mccutcheon2016_Experimental}, we devise multipartite entanglement distillation protocols that directly use GHZ states as the initial resource state. We are inspired by the scheme of Rengaswamy et al., which uses an $[[n,k,d]]$ stabilizer code to distill $n$ GHZ states into $k$ GHZ states in the presence of Pauli noise \cite{rengaswamy2022_Distilling,rengaswamy2024_Entanglement}. The first variant of the scheme is explicitly formulated for $3$-qubit GHZ states, and it assumes that the nodes receiving the qubits of the GHZ states lie on a linear network topology, relying on high-weight stabilizer measurements of an induced $[[2n,k]]$ code as an intermediate step \cite{rengaswamy2022_Distilling}. Distillation of $N$-qubit GHZ states with the same machinery is outlined in \cite{rengaswamy2024_Entanglement}, and the generalization to arbitrary networks is discussed at a high level.

In this work, we use the stabilizer formalism to develop a similar scheme on a star network topology which naturally allows any number of parties, relies exclusively on local operations and classical communication (LOCC), and ensures that the code distance is the same for all parties. We reinterpret the derivation of the scheme proposed by Rengaswamy et al. through the lens of the stabilizer formalism in \cref{sec:ghz_map}. We provide numerical results from our scheme using various quantum error correction (QEC) codes, including the 5-qubit code, the toric code \cite{kitaev2003_Faulttolerant}, and the $[[144,12,12]]$ Bivariate Bicycle code -- known as the Gross Code \cite{bravyi2024_Highthreshold}. We derive a lower bound on the final GHZ state fidelity from the probability of successful QEC by all parties, and we show how results for any $N$ can be efficiently obtained from the results of single-party decoder Monte Carlo simulations. Finally, we use the stabilizer formalism to show that our protocol is not limited to GHZ states. In particular, we prove a sufficient condition on the logical operators of the stabilizer code that allows $n$ copies of \textit{any} $N$-qubit stabilizer state to be distilled to $k$ copies of the same stabilizer state by our protocol. We also show that a significantly weaker sufficient condition, satisfied by all commonly used codes, allows for distillation of any CSS stabilizer state with our protocol.

\section{Background and Notation}
\label{sec:bkgd}

\subsection{Pauli Operators}
\label{sec:bkgd_pauli}

Throughout the paper, we will work with elements of the $m$-qubit Pauli group $\mathcal{P}_m$, which are built from tensor products of the single qubit Pauli operators.

\begin{equation}
    \label{eqn:single_q_paulis}
    I = \begin{pmatrix}
        1 & 0 \\
        0 & 1
    \end{pmatrix}, \quad X = \begin{pmatrix}
        0 & 1 \\
        1 & 0
    \end{pmatrix}, \quad Z = \begin{pmatrix}
        1 & 0, \\
        0 & -1
    \end{pmatrix}, \quad Y = i XZ  = \begin{pmatrix}
        0 & -i \\
        i & 0
    \end{pmatrix}
\end{equation}

\begin{equation}
    \label{eqn:pauli_group}
    \mathcal{P}_m = \{ i^{\kappa} E_1 \otimes E_2 \otimes ... \otimes E_m : E_j \in \{I, X, Y, Z \} \ \kappa \in \{0, 1, 2, 3 \} \}
\end{equation}

A Pauli operator with no phase, corresponding to $\kappa = 0$ in \cref{eqn:pauli_group}, is Hermitian and unitary. Let $\mathcal{Q}_m$ denote the subset of the Pauli group containing all such operators. $\mathcal{Q}_m$ is not a group because it is not closed under multiplication.

\begin{equation}
    \label{eqn:pauli_no_phase_subset}
    \mathcal{Q}_m = \{ E_1 \otimes E_2 \otimes ... \otimes E_m : E_j \in \{I, X, Y, Z \} \} 
\end{equation}

We can write the operators in $\mathcal{Q}_m$ as $E(a, b)$, where $a = [a_1, a_2, ..., a_m], b = [b_1, b_2, ..., b_m] \in \mathbb{F}_2^m$ are two binary row vectors representing the support of $X$ and $Z$ respectively.

\begin{equation}
    \label{eqn:E_ab_def}
    E(a, b) := (i^{a_1 b_1} X^{a_1} Z^{b_1} ) \otimes (i^{a_2 b_2} X^{a_2} Z^{b_2} ) \otimes ... (i^{a_m b_m} X^{a_m} Z^{b_m} ) = i^{a b^T} \bigotimes_{j=1}^m (X^{a_j} Z^{a_j})
\end{equation}

For example, $E([0, 1, 1, 0], [0, 0, 1, 1]) = I \otimes X \otimes Y \otimes Z = I_1 X_2 Y_3 Z_4 = IXYZ$. \cref{lem:E_ab_properties} gives some properties of the $E(a,b)$ notation that are straightforward to show.

\begin{lemma}
\label{lem:E_ab_properties}
Useful properties of the $E(a, b)$ notation for Pauli operators.
    \begin{enumerate} [label=(\alph*), leftmargin=2em]
        \item Transpose: $E(a, b)^T = (-1)^{ab^T} E(a, b)$
        \item Product: $E(a,b) \cdot E(c,d) =  i^{bc^T - ad^T} E(a + c, b + d)$ where the sums $a+c$ and $c+d$ are performed modulo 4 and \cref{eqn:E_ab_def} is extended directly to $a, b \in \mathbb{F}_4^m$.
        \item Commutation: $E(a,b) \cdot E(c,d) = (-1)^{\langle [a, b], [c, d] \rangle_s} E(c,d) \cdot E(a,b)$, where 

        \begin{equation*}
            \langle [a, b], [c, d] \rangle_s := ad^T + bc^T \ \text{(mod 2)}
        \end{equation*}

        is called the symplectic inner product between $[a, b]$ and $[c, d]$ in $\mathbb{F}_2^{2m}$.
    \end{enumerate}
\end{lemma}

To label Pauli operators whose support spans multiple subsystems each with $n$ qubits, we concatenate bitstrings of length $n$ as $[x, y] = [x_1, ..., x_n, y_1, ..., y_n] \in \mathbb{F}_2^{2n}$. Concatenation of 3 bitstrings is written as $[x,y,z]$, and so on. Within this context, $0$ is taken to be the zero vector of length $n$. This concatenation convention is the reason for the party-major indexing discussed below in \cref{sec:bkgd_parties_copies}. If an operator is local to a single subsystem $A_i$, we can also write it as $E(a,b)_{A_i}$.

\subsection{Stabilizer Codes}
\label{sec:bkgd_stab_codes}

Stabilizer codes are a widely studied class of qubit QEC codes \cite{gottesman1997_Stabilizer}. A stabilizer code encoding $k$ logical qubits into the collective state of $n$ physical qubits is described by a stabilizer group $\mathcal{C}$ generated by $r := n - k$ independent, mutually commuting Pauli operators $c_1, c_2, ..., c_r \in \mathcal{P}_n$ that we will refer to as ``code generators". The codespace of the code specified by $\mathcal{C}$ is the simultaneous $+1$ eigenspace of all the code generators. The stabilizer group $\mathcal{C}$ can be represented compactly by its parity check matrix $H_{\mathcal{C}}$, whose $r$ rows specify the code generators as shown in \cref{eqn:gen_from_row}.

\begin{equation}
    \label{eqn:code_parity_check}
    H_{\mathcal{C}} = \left( \begin{array}{c|c|c}
        X & Z & \alpha
    \end{array} \right) \in \mathbb{F}_2^{r \times (2n + 1)}
\end{equation}

$X$ and $Z$ here are $r \times n$ binary matrices specifying the qubits on which Pauli X and Z have non-trivial support, and $\alpha \in \mathbb{F}_2^{r}$ is a column vector of ``sign bits". Let $x_i$ denote the $i$th row of $X$, and $z_i$ similar. Then the generator $c_i$ is obtained from its row in $H_{\mathcal{C}}$ as below.

\begin{equation}
    \label{eqn:gen_from_row}
    c_i = (-1)^{\alpha_i} E(x_i, z_i)
\end{equation}

Commutativity of the generators is ensured by the relation $XZ^T + ZX^T = 0$ (mod 2), and the generators are independent if and only if $\text{rank} ( X | Z ) = r$.

The logical operators of a stabilizer code are Pauli operators on the $n$ physical qubits that commute with all of the code generators. They are members of $N(\mathcal{C})$, the normalizer of $\mathcal{C}$ in $\mathcal{P}_n$. The logical Pauli group of the code is the quotient group $N(\mathcal{C}) / \mathcal{C} \cong \mathcal{P}_k$. In \cref{eqn:logical_mapping}, we introduce the mapping $\pi_{\mathcal{C}}$ from a physical operation $g \in N(\mathcal{C})$ to the coset in $N(\mathcal{C}) / \mathcal{C}$ that contains $g$, and therefore the realized logical Pauli by isomorphism.

\begin{equation}
    \label{eqn:logical_mapping}
    \pi_{\mathcal{C}} : N(\mathcal{C}) \rightarrow N(\mathcal{C}) / \mathcal{C} \cong \mathcal{P}_{k}
\end{equation}

Logical operators are typically denoted with a bar. For example $\overline{X_1} = X^{\otimes 5} \in \mathcal{P}_5$ is the operator that acts on the physical qubits of the 5-qubit code to enact Pauli X on the 1st (and only) logical qubit. Considering this definition, we see that the mapping $\pi_{\mathcal{C}}$ removes the bar, i.e. $\pi_{\mathcal{C}} (\overline{X_1}) = X_1 \in \mathcal{P}_1$. We can extend this to an arbitrary code with $k$ logical qubits using the $E(a,b)$ notation.

\begin{equation}
    \label{eqn:logical_mapping_Eab}
    \pi_{\mathcal{C}} \left( \overline{E(a,b)} \right) = E(a,b) \quad \quad a, b \in \mathbb{F}_2^k
\end{equation}

\subsection{Stabilizer States}
\label{sec:bkgd_stab_states}

An $N$-qubit stabilizer state is the simultaneous $+1$ eigenstate of exactly $N$ mutually commuting, independent $N$-qubit Pauli operators \cite{gottesman1998_Heisenberg}. Describing a state in this manner can be interpreted as defining a stabilizer code with $k=0$ logical qubits so that the $2^k$-dimensional code space specifies a single ray in the $N$-qubit Hilbert space.

In this paper, we write the stabilizer group specifying a general $N$-qubit stabilizer state as $\mathcal{S} = \langle s_1, s_2, ... , s_N \rangle$, and we refer to the $s_i \in \mathcal{P}_N$ as ``state generators". We can again work with $\mathcal{S}$ via its parity check matrix $H_{\mathcal{S}} \in \mathbb{F}_2^{N \times (2N+1)}$. Extending the notation of the previous section, the submatrices $A$ and $B$ are $N \times N$ binary matrices, and $\tau$ is a length-$N$ vector of sign bits.

\begin{equation}
    \label{eqn:state_parity_check}
    H_{\mathcal{S}} = \left( \begin{array}{c|c|c}
        A & B & \tau
    \end{array} \right)
\end{equation}

\begin{equation}
    \label{eqn:state_gen_from_row}
    s_i = (-1)^{\tau} E(a_i, b_i)
\end{equation}

Again commutativity of the state generators is satisfied by ensuring $AB^T + BA^T = 0$, and independence is guaranteed by $\text{rank} ( A|B ) = N$.

When a stabilizer state undergoes a unitary evolution or a measurement of a Pauli operator, the resulting change in the state can be tracked by updating the state generators $s_1, s_2, ..., s_N$ according to well-known rules \cite{gottesman1998_Heisenberg}, which we give in \cref{sec:stab_rules} for completeness.

GHZ states are stabilizer states, and the distillation scheme developed in this work relies on Pauli measurements and Pauli operations (noise and correction). The stabilizer formalism therefore provides physical insight into the collective state shared by the parties in the network, as well as an efficient method of numerically simulating the evolution of the state throughout the scheme. 

\subsection{Parties and Copies}
\label{sec:bkgd_parties_copies}

Throughout this paper, we imagine a setting where there are $N$ parties (or nodes, or subsystems) that each have $n$ qubits in some collective stabilizer state. We label the parties as $A_1, A_2, ..., A_N$ in general, but in the case of $N=3$ we adopt the convention of referring to $A_1$ as Alice, $A_2$ as Bob, and $A_3$ as Charlie. We can track the evolution of the collective state via its parity check matrix $H_{tot} \in \mathbb{F}_2^{nN \times (2nN + 1)}$ (and the stabilizer group $\mathcal{S}_{tot}$ that it generates). To be concrete, we adopt a party-major indexing for the rows of $H_{tot}$, so that columns $1-n$ of the X and Z support submatrices refer to the $A_1$ qubits, columns $(n+1)-(2n)$ to the $A_2$ qubits, and so on.

There are two types of collective states that deserve a shorthand notation. Consider the situation in which each party's $n$ qubits are in a codespace of $\mathcal{C}$ specified by the sign bits $\alpha^{(i)} \in \mathbb{F}_2^r$. Then, with this indexing convention, the parity check matrix for the overall code has a direct sum structure.

\begin{equation}
    \label{eqn:party_direct_sum}
    H_{\mathcal{C}}^{\oplus N}
    :=
    \left(
    \begin{array}{cccc|cccc|c}
        X &   &        &   & Z &   &        &   & \alpha^{(1)} \\
            & X &        &   &   & Z &        &   & \alpha^{(2)} \\
            &   & \ddots &   &   &   & \ddots &   & \vdots \\
            &   &        & X &   &   &        & Z & \alpha^{(N)}
    \end{array}
    \right)
    \in \mathbb{F}_2^{rN \times (2nN+1)} .
\end{equation}

We adopt the shorthand that $\mathcal{C}^{\oplus N}$ is the stabilizer group on $nN$ qubits that is generated by the rows of $H_{\mathcal{C}}^{\oplus N}$. You will notice that $H_{\mathcal{C}}^{\oplus N}$ is $kN$ rows short of specifying a single stabilizer state on the $nN$ qubits because the $k$ logical qubits at each party are unconstrained. We can add $kN$ more generators to arrive at a stabilizer state described by some $H_{tot}$. Since all of the rows of $H_{tot}$ must produce mutually commuting and independent Pauli operators, each of the $kN$ additional rows must represent some tensor product of logical operators on each party. Introducing $H_{\mathcal{L}} \in \mathbb{F}_2^{kN \times (2nN +1)}$ such that its $kN$ rows specify these products of logical operators, we say that $H_{tot} = H_{\mathcal{C}}^{\oplus N} \cup H_{\mathcal{L}}$ is the parity check matrix for this logical stabilizer state. We use $\cup$ to denote appending the rows of $H_{\mathcal{L}}$ to the bottom of $H_{\mathcal{C}}^{\oplus N}$ because we are effectively taking the union of two sets of generators. This operation only produces a well-defined stabilizer group when the matrices $H_{\mathcal{C}}^{\oplus N}$ and $H_{\mathcal{L}}$ have the same number of columns, and their rows produce mutually commuting Pauli operators.

\begin{figure}[h]
    \centering
    \includegraphics[width=0.5\linewidth]{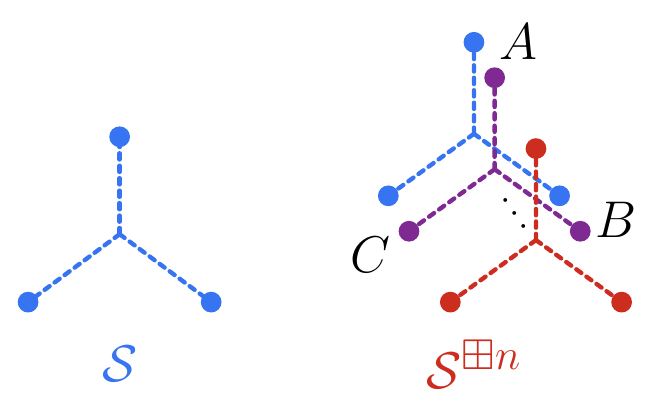}
    \caption{On the left, an arbitrary 3-qubit stabilizer state given by the stabilizer group $\mathcal{S}$ is depicted. The dots represent the qubits. On the right, we show how $n$ independent copies of this state can be arranged such that three parties (Alice, Bob, and Charlie) each have $n$ qubits, with one qubit from each copy. The overall stabilizer group for the $n$ copies is $\mathcal{S}^{\copymajorsum n}$.}
    \label{fig:stab_state_ncopies}
\end{figure}

Now consider the situation in which there are $n$ copies of the $N$-qubit stabilizer state $\mathcal{S}$, and we arrange the qubits such that $A_j$ has the $j$th qubit from each copy. This arrangement is depicted for $N=3$ in \cref{fig:stab_state_ncopies}. Since the $n$ copies are independent of each other, the parity check matrix $H_{tot}$ again has a direct sum structure, but with respect to a copy-major indexing. We introduce the symbol $\copymajorsum$ to denote this copy-major direct sum, and we use $\mathcal{S}^{\copymajorsum n}$ to refer to the stabilizer group (equivalently, the stabilizer state) specified by $H_{\mathcal{S}}^{\copymajorsum n}$.

\begin{example}
    Clarifying the party-major and copy-major direct sums with an explicit $n=N=2$ example.
\end{example}

Consider an $N=2$-qubit stabilizer state with the parity check matrix in \cref{eqn:state_parity_check}. We explicitly enumerate the entries of the $A$ and $B$ submatrices, and we suppose that the vector of sign bits is $\tau = 0$.

\begin{equation*}
    A = \begin{pmatrix}
        a_{11} & a_{12} \\
        a_{21} & a_{22}
    \end{pmatrix} \quad \quad \text{and} \quad \quad B = \begin{pmatrix}
        b_{11} & b_{12} \\
        b_{21} & b_{22}
    \end{pmatrix}
\end{equation*}

To describe the situation where there are $n=2$ copies of the state arranged such that $A_1$ has the first qubit from each copy and $A_2$ has the second, we use the copy-major direct sum $\copymajorsum$ to construct the overall parity check matrix for the $nN=4$-qubit stabilizer state, which we write down in equation \cref{eqn:copy_major_example}. The first two columns of each $4 \times 4$ submatrix refer to the $A_1$ qubits, and the second two columns to the $A_2$ qubits.

\begin{equation}
    \label{eqn:copy_major_example}
    H_{\mathcal{S}}^{\copymajorsum 2} = \left( \begin{array}{cccc|cccc|c}
        a_{11} & 0 & a_{12} & 0 & b_{11} & 0 & b_{12} & 0 & 0 \\
        0 & a_{11} & 0 & a_{12} & 0 & b_{11} & 0 & b_{12} & 0 \\
        a_{21} & 0 & a_{22} & 0 & b_{21} & 0 & b_{22} & 0 & 0 \\
        0 & a_{21} & 0 & a_{22} & 0 & b_{21} & 0 & b_{22} & 0
    \end{array} \right)
\end{equation}

Meanwhile, the party-major direct sum $\oplus$ describes the alternative situation in which the 2 qubits of $A_1$ are in the state $\mathcal{S}$, and the 2 qubits of $A_2$ are also (separately) in the state $\mathcal{S}$. This looks like a standard direct sum, since we use a party-major indexing convention when we explicitly write down the parity check matrix for a multipartite stabilizer state.

\begin{equation}
    \label{eqn:party_major_example}
    H_{\mathcal{S}}^{\oplus 2} = \left( \begin{array}{cccc|cccc|c}
        a_{11} & a_{12} & 0 & 0 & b_{11} & b_{12} & 0 & 0 & 0 \\
        a_{21} & a_{22} & 0 & 0 & b_{21} & b_{22} & 0 & 0 & 0 \\
        0 & 0 & a_{11} & a_{12} & 0 & 0 & b_{11} & b_{12} & 0 \\
        0 & 0 & a_{21} & a_{22} & 0 & 0 & b_{21} & b_{22} & 0
    \end{array} \right)
\end{equation}

\subsection{GHZ States}
\label{sec:bkgd_ghz}

The $N$-qubit GHZ state shared between parties $A_1, A_2, ..., A_N$ is

\begin{equation}
    \label{eqn:ghz_state}
    \ket{\text{GHZ}^{(N)}} = \frac{\ket{0}_{A_1} \ket{0}_{A_2} ... \ket{0}_{A_N} + \ket{1}_{A_1} \ket{1}_{A_2} ... \ket{1}_{A_N}}{\sqrt{2}}.
\end{equation}

This state is a stabilizer state, and we will denote its stabilizer group as $\mathcal{G}_N$ throughout the paper. One of its generators is a global $XX...X$ operator, and the remaining $(N-1)$ generators are $ZZ$ operators on adjacent qubits. The parity check matrix for $\mathcal{G}_4$ is given below as an example.

\begin{equation}
    \label{eqn:ghz_parity_check}
    H_{\mathcal{G}_4} = \left( \begin{array}{cccc|cccc|c}
        1 & 1 & 1 & 1 & 0 & 0 & 0 & 0 & 0 \\
        0 & 0 & 0 & 0 & 1 & 1 & 0 & 0 & 0 \\
        0 & 0 & 0 & 0 & 0 & 1 & 1 & 0 & 0 \\
        0 & 0 & 0 & 0 & 0 & 0 & 1 & 1 & 0 \\
    \end{array} \right)
\end{equation}

We will be working with $n$ copies of the GHZ state, i.e. $\ket{\text{GHZ}_n^{(N)}} := \ket{\text{GHZ}^{(N)}}^{\otimes n}$. This state is described in our notation by the stabilizer group $\mathcal{G}_N^{\copymajorsum n}$. It will sometimes be helpful to explicitly refer to the Pauli operators that generate this group. There are $n$ independent sets of generators corresponding to each copy, which we index by $j$ in \cref{eqn:main_ghz_stab_Nn} below.

\begin{equation}
\label{eqn:main_ghz_stab_Nn}
    \mathcal{G}_N^{\copymajorsum n} = \langle Z_{(A_1)j} Z_{(A_2)j}, Z_{(A_2)j} Z_{(A_3)j}, \hdots, Z_{(A_{N-1})j} Z_{(A_N)j}, X_{(A_1)j} X_{(A_2)j} \hdots X_{(A_N)j} : j = 1, 2, ..., n \rangle
\end{equation}

\section{GHZ Distillation on Star Topology}
\label{sec:scheme}

\begin{figure}[h]
    \centering
    \includegraphics[width=0.6\textwidth]{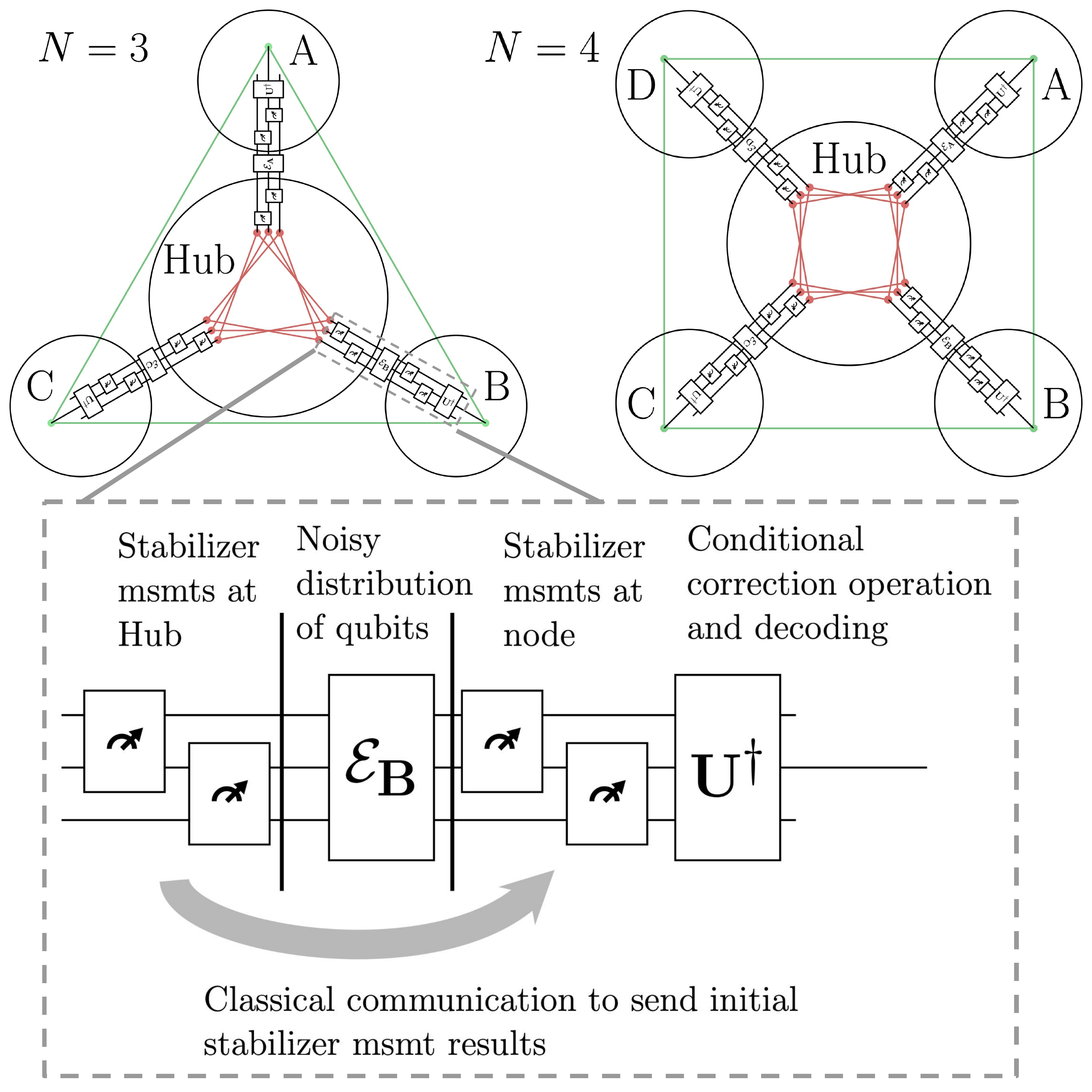}
    \caption{A visual depiction of the $N$-qubit GHZ distillation scheme on a star topology network for $N=3$ and $N=4$. Dots represent physical qubits, while an $N$-qubit GHZ state is drawn as a regular $N$-sided polygon connecting $N$ dots, i.e. triangles in the $N=3$ figure and squares in the $N=4$ figure. Both figures depict an example beginning with three copies at the central hub (small red polygons). After the circuit is complete, each party has one qubit in a distributed GHZ state shared between all other parties (large green polygon). Time proceeds radially outward from the hub, and circuit operations inside a particular circle are performed on that party's processor. Qubits experience a party-specific Pauli noise channel $\mathcal{E}$ en route from the Hub to the end party's processor.}
    \label{fig:scheme}
\end{figure}

We now give the details of our scheme which uses an $[[n,k,d]]$ stabilizer code to prepare $k$ GHZ states. Initially, $n$ copies of the $N$-qubit GHZ state are prepared at a central hub and arranged such that each of the $N$ parties has one qubit from each of the $n$ copies. This arrangement is visualized in \cref{fig:scheme} for $N=3,4$. While the qubits are still at the hub, the code generators $c_1, c_2, ..., c_r \in \mathcal{P}_n$ are measured separately on each party's qubits, resulting in a random bitstring of measurement outcomes $\alpha^{(i)} \in \mathbb{F}_2^r$ for each party. These measurement results are sent over a classical channel to each party. We assume for now that local operations and classical communication are noiseless. All of the code generators commute with each other, so the initial measurement step at the hub can happen in parallel. After this round of measurements, the resulting state is sent over noisy (Pauli) channels $\mathcal{E}_{A_i}$ to the parties, which are assumed to be spatially separated from the hub. In practice, this step can happen either by physically shuttling the qubits, or by teleporting the state with Bell measurements on $n$ imperfect Bell pairs shared between each party and the hub. This noisy distribution step may result in a Pauli error on each party's qubits, so the parties perform another round of code generator measurements to determine the corresponding error syndrome and apply a correction if necessary. The scheme fails when the correction step introduces a logical error that affects the encoded GHZ state.

\begin{algorithm}
\caption{One-way scheme for distilling $N$-qubit GHZ states shared by $N$ parties each connected to a central network hub.}
\label{alg:scheme}
\begin{algorithmic}
\Require $n$ copies of the $N$-qubit GHZ state at the central hub, $[[n,k,d]]$ stabilizer code specified by a list of code generators $c_1, c_2, ..., c_r$ with an agreed-upon order. ($r:=n-k$)
\Ensure $k$ copies of the $N$-qubit GHZ state shared by the $N$ parties, unless the QEC step introduces a logical error that is not in the stabilizer group $\mathcal{G}_N^{\copymajorsum k}$.

\Statex \hrulefill

\State \textbf{Central Hub}
\algstep{(a)}{Arranges qubits such that each party $A_1, A_2, ..., A_N$ receives one qubit from each of the $n$ copies of the GHZ state.}
\algstep{}{\textbf{For} $\mathbf{i=1,2,...,N}$\textbf{:}}
\algstepind{(b,i)}{Measures the code generators $c_1, c_2, ..., c_r$ on the $A_i$ qubits, obtaining a vector of measurement results $\alpha^{(i)} \in \mathbb{F}_2^r$.}
\algstepind{(c,i)}{Sends $A_i$ their qubits over a noisy (Pauli) channel $\mathcal{E}_{A_i}$. Also sends the measurement results $\alpha^{(i)}$ over a perfect classical channel.}

\State
\State

\State \textbf{Each Party $A_i$, after receiving their qubits}
\algstep{(a)}{Measures the code generators $c_1, c_2, ..., c_r$ on their qubits, obtaining a vector of measurement results $\beta^{(i)} \in \mathbb{F}_2^r$.}
\algstep{(b)}{Calculates the error syndrome $\gamma^{(i)} = \alpha^{(i)} \oplus \beta^{(i)}$ and uses $\gamma^{(i)}$ to infer the Pauli error applied during the noisy distribution step.}
\algstep{(c)}{Performs a recovery operation on their qubits, if necessary, to return to the $\alpha^{(i)}$ eigenspace of the code generators that the state was originally projected into at the hub.}
\algstep{(d)}{\textit{Optional:} Inverts the encoding unitary to obtain $k$ physical qubits in GHZ states shared with the other parties. Not necessary if further computations are to be performed at the logical level.}

\State
\State

\State // See \cref{sec:fidelity} for a discussion of when the scheme is successful and when it fails

\end{algorithmic}
\end{algorithm}

Our scheme is a one-way GHZ distillation protocol that reduces to local QEC by each party upon receipt of their qubits. Furthermore, the steps for each party can be conducted in parallel, and the generalization to arbitrary $N$ adds no conceptual complexity. The latter point is illustrated by the change between $N=3$ and $N=4$ in \cref{fig:scheme}: simply add another spoke to the wheel. The tradeoff for this simplicity is that a maximum of $rN$ stabilizer measurements must be performed at the hub prior to distributing the qubits. This may be prohibitive for large $N$ and $n$, but any optimization of syndrome measurement circuits can be applied to optimize this scheme for a particular code. In addition, due to the form of the Z-type generators of $\mathcal{G}_N$, any code generator consisting purely of Pauli Z operators only needs to be measured on one subsystem. This significantly reduces the number of measurements that must be performed for CSS codes (and any other code that involves some purely $Z$ code generators), especially for large $N$.

The classical communication requirements of our scheme are modest. While the hub must send $rN$ total bits of classical information per round of distillation, each party $A_i$ only needs to receive $r$ bits containing the measurement results $\alpha^{(i)}$ from their own qubits. Then, they perform local QEC on their qubits according to the error syndrome $\gamma^{(i)} = \alpha^{(i)} \oplus \beta^{(i)}$, and no classical communication between parties is required.

\begin{theorem}
\label{thm:Z_spec}
    Suppose the code generator $E(0, b)_{A_1}$ is measured on the $A_1$ subsystem, obtaining the result $\mu = \pm 1$. Then the state $\ket{\text{GHZ}_n^{(N)}}_{A_1 A_2 ... A_N}$ is simultaneously projected into the $\mu$ eigenspace of $E(0,b)_{A_k}$ for all $k = 2, ..., N$.
\end{theorem}

\begin{proof}
    Recall the stabilizer group $\mathcal{G}_N^{\copymajorsum n}$ for $n$ copies of the $N$-qubit GHZ state, defined in \cref{eqn:main_ghz_stab_Nn}. Since the measured stabilizer $E(0, b)_{A_1}$ has no Pauli X terms on any qubit, it commutes with all of the $ZZ$ generators of $\mathcal{G}_N^{\copymajorsum n}$.

    \begin{equation}
        \label{eqn:comm_with_all_ZZ}
        [E(0,b)_{A_1}, Z_{(A_j)i} Z_{(A_{j+1})i}] = 0 \quad \forall i, j
    \end{equation}

    When $E(0,b)_{A_1}$ is measured and the outcome $\mu = \pm 1$ obtained, the stabilizer group is updated to $\mathcal{G}'$, where $\mu E(0,b)_{A_1} \in \mathcal{G}'$ (among other updates) \cite{gottesman1998_Heisenberg}. Only generators from the original stabilizer group which anticommute with the measured operator are removed or altered. Since the ZZ stabilizers all commute with the measured operator, we have $Z_{(A_j)i} Z_{(A_{j+1})i} \in \mathcal{G}'$ for all $i, j$.

    By closure of the stabilizer group, we can multiply the $ZZ$ generators on adjacent subsystems to obtain the further result $Z_{(A_1)i} Z_{(A_k)i} \in \mathcal{G}'$ for all $k = 2, ..., N$. This will now allow us to express $\mu E(0,b)_{A_k}$ as a product of elements of the stabilizer group for any $k=2,...,N$, which concludes the proof.

    \begin{equation}
        \label{eqn:op_in_stab}
        (\mu E(0,b)_{A_1}) \cdot \prod_{i = 1}^n \left( Z_{(A_1)i} Z_{(A_k)i} \right)^{b_i} = \mu E(0,b)_{A_k}
    \end{equation}
    
\end{proof}

\subsection{Stabilizer Tracking and Output Fidelity}
\label{sec:fidelity}

Initially, the collective state is $n$ copies of the $N$-qubit GHZ state, described by the stabilizer group $\mathcal{G}_N^{\copymajorsum n}$ with parity check matrix $H_{tot} = H_{\mathcal{G}_N}^{\copymajorsum n}$. Then, the code generators are measured on each party's qubits. The parity check matrix for the resulting stabilizer group is 

\begin{equation}
    \label{eqn:H_tot_prime}
    H_{tot}' = H_{\mathcal{C}}^{\oplus N} \cup H_{\mathcal{L}} .
\end{equation}

Whenever $\mathcal{C}$ admits a basis of logical operators in which the logical X (Z) operators consist of purely Pauli X (Z) operations on the physical qubits, the non-local state generators given by the rows of $H_{\mathcal{L}}$ are logical operators that stabilize the $k$ logical copies of the GHZ state (\cref{thm:css_states}). The purpose of the QEC step at each party is to return the code generators and these logical state generators to the correct signs, after the state is corrupted by Pauli noise in the distribution step. This combined noise and correction step results in the application of some Pauli operator to the overall system that commutes with all the code generators. The final state therefore has parity check matrix

\begin{equation}
    \label{eqn:H_tot_prime_prime}
    H_{tot}'' = H_{\mathcal{C}}^{\oplus N} \cup H_{\mathcal{L}}' .
\end{equation}

$H_{\mathcal{C}}^{\oplus N}$ remains unchanged between $H_{tot}'$ and $H_{tot}''$, since the correction step returns each party to the $\alpha^{(i)}$ codespace that was originally observed at the hub. However, if the combined noise and correction operation anticommutes with any of the logical stabilizers in $H_{\mathcal{L}}$, we will have $H_{\mathcal{L}}' \neq H_{\mathcal{L}}$. We now work through this discussion of the noise and correction step in more detail, and show how the fidelity of the protocol can be calculated for a particular code and decoder from the single-party QEC performance of the decoder. After the initial measurements are performed and the state is described by $H_{tot}'$, Pauli channels between the hub and each party act on the state with the error operator

\begin{equation}
    \label{eqn:error_op}
    D = D_1 \otimes D_2 \otimes ... \otimes D_N \quad \quad D_i \in \mathcal{P}_n.
\end{equation}

The subsequent round of syndrome measurements allows each party to infer the error $D_i$ that affected their qubits (up to a trivial global phase). This inference is performed by a decoding algorithm specific to the QEC code being employed. Each party then independently applies a recovery operation $R_i$ determined by their decoder output.

\begin{equation}
    \label{eqn:recov_op}
    R = R_1 \otimes R_2 \otimes ... \otimes R_N \quad \quad R_i \in \mathcal{P}_n
\end{equation}

The combined operation is $R D = \bigotimes_i R_i D_i$, where each $R_i D_i \in N(\mathcal{C})$ is some (possibly trivial) logical operator of the code $\mathcal{C}$. The operation enacted by each $R_i D_i = \overline{E(a_i,b_i)}$ on the $k$ logical qubits of the $A_i$ subsystem is obtained via the mapping $\pi_{\mathcal{C}}$ defined in \cref{eqn:logical_mapping_Eab}.

\begin{equation}
    \label{eqn:logical_noise_corr_Ai}
    \pi_{\mathcal{C}} (R_i D_i)_{A_i} =  E(a_i, b_i)_{A_i}, \quad \quad a_i, b_i \in \mathbb{F}_2^k
\end{equation}

The overall operation applied to the collective logical state is a $kN$-qubit Pauli operator $L$, up to a global phase $i^{\kappa}$ which we ignore.

\begin{equation}
    \label{eqn:logical_noise_corr_all}
    L = \pi_{\mathcal{C}}^{\oplus N} (RD) = \bigotimes_{i=1}^N E(a_i, b_i)_{A_i} \in \mathcal{P}_{kN}.
\end{equation}

The scheme is successful if and only if $L \in \mathcal{G}_N^{\copymajorsum k}$, i.e. $L$ is in the stabilizer group for $k$ copies of the noiseless GHZ state (modulo the global phase). If $L \notin \mathcal{G}_N^{\copymajorsum k}$, then the noise and correction step take the logical GHZ state described by $H_{tot}'$ to an orthogonal state. Let the probability of a particular logical error $E(a,b)_{A_i}$ occurring on the $i$th subsystem be $p_i (a,b)$. In practice, $p_i$ is set by the noise channel $\mathcal{E}_{A_i}$ between the hub and $A_i$, as well as the decoding algorithm used by $A_i$ to infer $R_i$ from the error syndrome $\gamma^{(i)}$. The probability of a particular overall logical error $L$ with the form in \cref{eqn:logical_noise_corr_all} is therefore

\begin{equation}
    \label{eqn:prob_L}
    P(L) = \prod_{i=1}^N p_i (a_i, b_i).
\end{equation}

We can now explicitly characterize the output mixed state of the protocol, $\rho_{\text{out}}$. Any single realization will produce a pure state drawn from the classical ensemble described by \cref{eqn:output_mixed_state}.

\begin{equation}
    \label{eqn:output_mixed_state}
    \begin{split}
    \rho_{\text{out}} &= \sum_{L \in \mathcal{Q}_{kN}} P(L) \cdot L \ket{\text{GHZ}_k^{(N)}} \bra{\text{GHZ}_k^{(N)}} L  \\
    &= \sum_{L \in \mathcal{G}_N^{\copymajorsum k}} P(L) \ket{\text{GHZ}_k^{(N)}} \bra{\text{GHZ}_k^{(N)}} + \sum_{L \notin \mathcal{G}_N^{\copymajorsum k}} P(L) \cdot L \ket{\text{GHZ}_k^{(N)}} \bra{\text{GHZ}_k^{(N)}} L
    \end{split}
\end{equation}

From this form of $\rho_{\text{out}}$ we can see that the fidelity of the protocol with the desired output pure state $\ket{\text{GHZ}_k^{(N)}}$ is equal to the probability of applying an overall logical error that is a stabilizer of the GHZ state.

\begin{equation}
    \label{eqn:fidelity}
    f_N = \sum_{L \in \mathcal{G}_N^{\copymajorsum k}} P(L)
\end{equation}

We work through details of the calculation for $f_N$ from the single-party logical error probability distributions $p_i(a_i, b_i)$ in \cref{sec:fidelity_details}. One particularly useful special case emerges when every party's noise channel and decoder are identical. Then, each party's logical error probability distribution is the same, $p_i (a, b) = p (a, b)$ for all $i=1,2, \hdots, N$. Let $\tilde{f_N}$ denote the output fidelity in this case. 

\begin{equation}
    \label{eqn:fidelity_equal}
    \tilde{f_N} = \frac{1}{2^k} \sum_{a, s \in \mathbb{F}_2^k} q(a,s)^N ,
\end{equation}

Where $q(a,s)$ in \cref{eqn:fidelity_equal} is defined as

\begin{equation}
    q(a, s) := \sum_{b \in \mathbb{F}_2^k} (-1)^{s b^T} p (a, b).
\end{equation}

This quantity arises from enforcing the global Z parity constraint $b_1 \oplus b_2 \oplus ... \oplus b_N = 0$ that is obeyed by all stabilizers of the GHZ state. See \cref{sec:fidelity_details} for the full details of derivation. Notably, \cref{eqn:fidelity_equal} allows us to calculate the fidelity for arbitrarily large $N$ without incurring the exponential cost of enumerating every overall logical error $L$ on $N$ subsystems. In fact, even in the worst case where each party's logical error probability distribution $p_i$ is unique, our method of calculating $p_{\text{succ}}^{(N)}$ scales only linearly in $N$, which is an exponential improvement over naively iterating through all possible logical errors. We take advantage of this simplification to calculate output infidelities for up to $N=100$, which we show in \cref{fig:5qubit_results}.

In summary, we showed in this section that our scheme can tolerate some logical errors in the single-party QEC step, if those logical errors conspire to stabilize the GHZ state. There is always a finite probability of such a fortuitous logical error, which makes the probability of perfect correction by all parties a lower bound for the fidelity of the logical GHZ state at the end of the protocol. If we suppose for simplicity that all parties experience the same noise channel, each party will have the same logical error rate $p_{\text{LE}}$, which gives the probability of perfect correction the form given in \cref{eqn:fid_lower_bound}.

\begin{equation}
    \label{eqn:fid_lower_bound}
    \tilde{f_N} \geq ( 1 - p_{\text{LE}})^N .
\end{equation}

The bound becomes tight for small $p_{\text{LE}}$, since the fortuitous logical errors that stabilize the GHZ state require at least two parties to fail in their QEC step. This is confirmed by the results in \cref{fig:5qubit_bound}.

\subsection{Code Capacity Results}
\label{sec:scheme_cc_results}

\begin{figure}[h]
    \centering

    \begin{minipage}[c]{0.4\textwidth}
        \centering
        \begin{subfigure}{\linewidth}
            \centering
            \includegraphics[width=\linewidth]{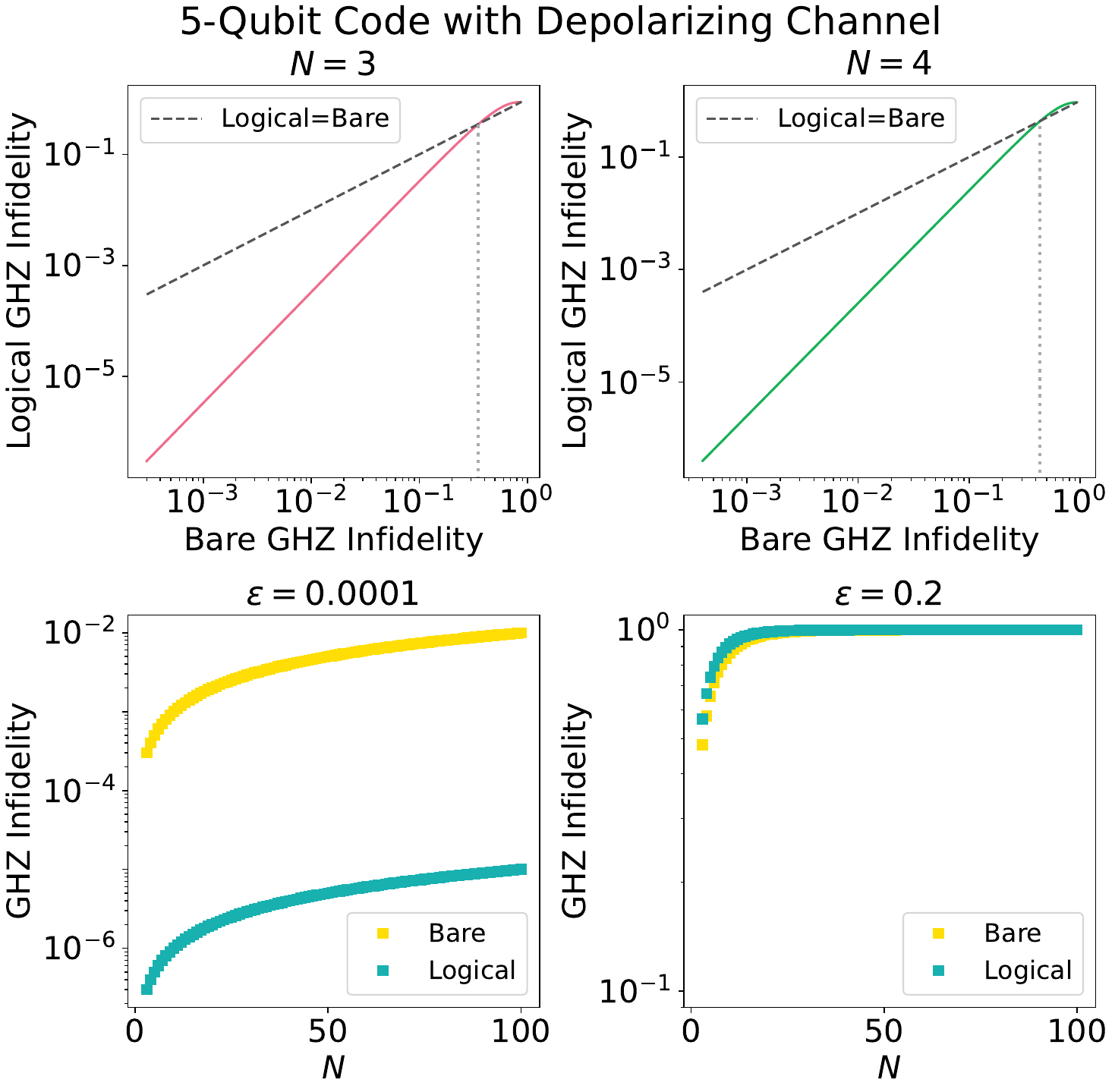}
            \caption{}
            \label{fig:5qubit_results}
        \end{subfigure}
    \end{minipage}
    \hspace{3mm}
    \begin{minipage}[c]{0.4\textwidth}
        \centering

        \begin{subfigure}{\linewidth}
            \centering
            \includegraphics[width=\linewidth]{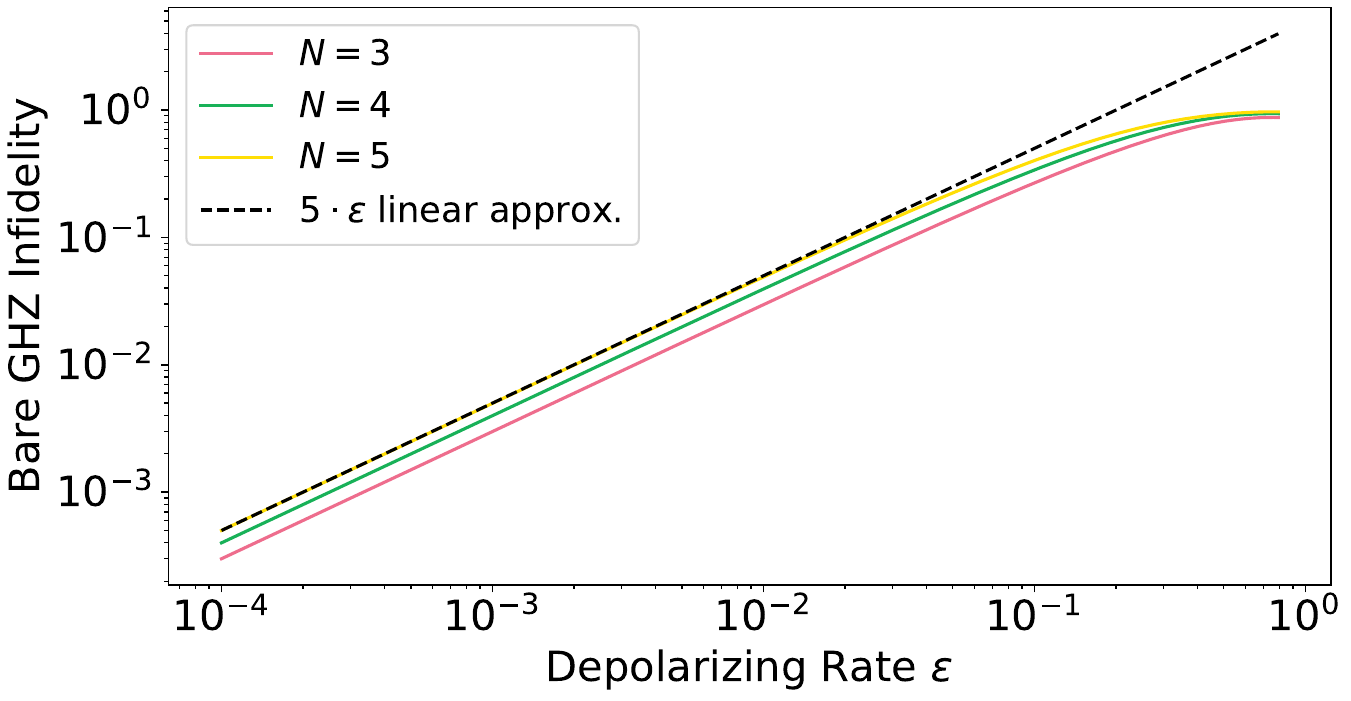}
            \caption{}
            \label{fig:bare_vs_eps}
        \end{subfigure}

        \vspace{0.5em}

        \begin{subfigure}{\linewidth}
            \centering
            \includegraphics[width=\linewidth]{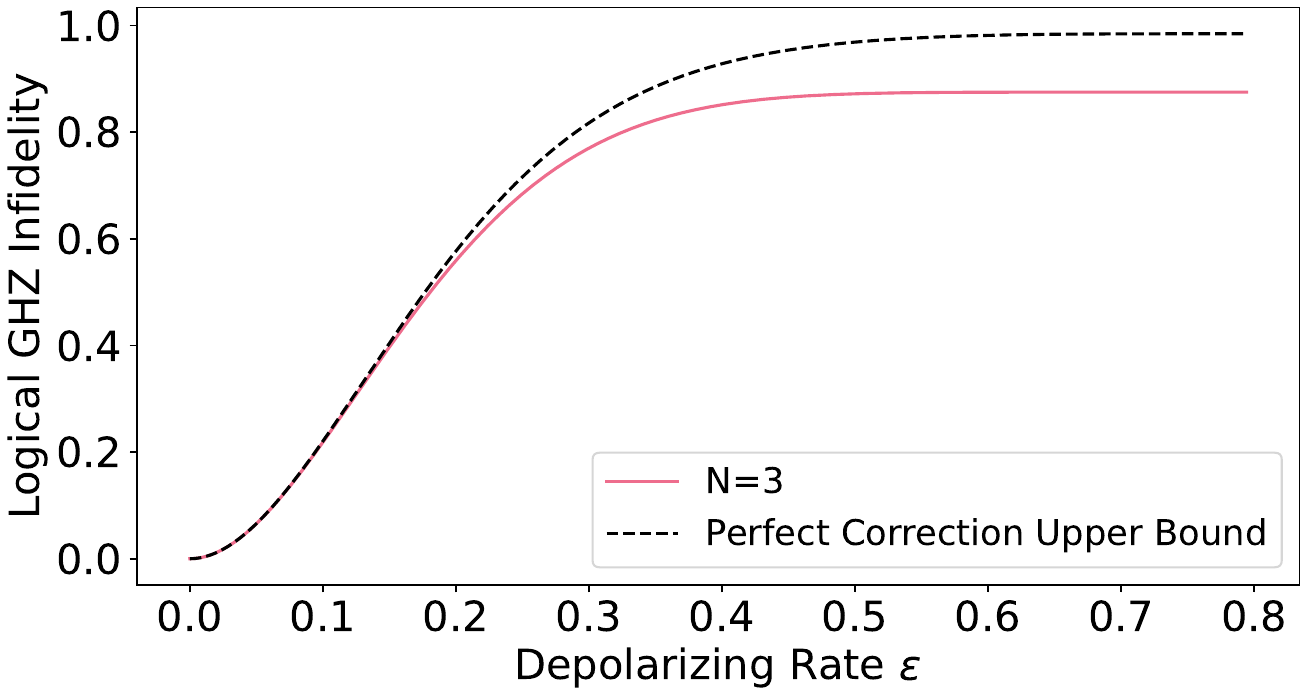}
            \caption{}
            \label{fig:5qubit_bound}
        \end{subfigure}
    \end{minipage}

    \caption{Results for GHZ distillation with the 5-qubit code in the presence of depolarizing noise. The top row of Figure (a) shows the output infidelity of the scheme for $N=3$ and $4$ plotted against the bare infidelity discussed in the main text. The bottom row gives results up to large $N$, calculated from \cref{eqn:fidelity_equal}, for two choices of the depolarizing rate $\varepsilon$. Figure (b) shows the bare infidelity as a function of $\varepsilon$ for $N=3,4,5$ which, for small $\varepsilon$, is well-approximated by $N \cdot \varepsilon$, since weight-1 errors dominate. Figure (c) shows, for $N=3$, how the logical infidelity collapses to the bound in \cref{eqn:fid_lower_bound} for small $\varepsilon$.}
    \label{fig:5qubit}
\end{figure}

\begin{figure}
    \centering

    \begin{subfigure}{0.38\textwidth}
        \centering
        \includegraphics[width=\linewidth]{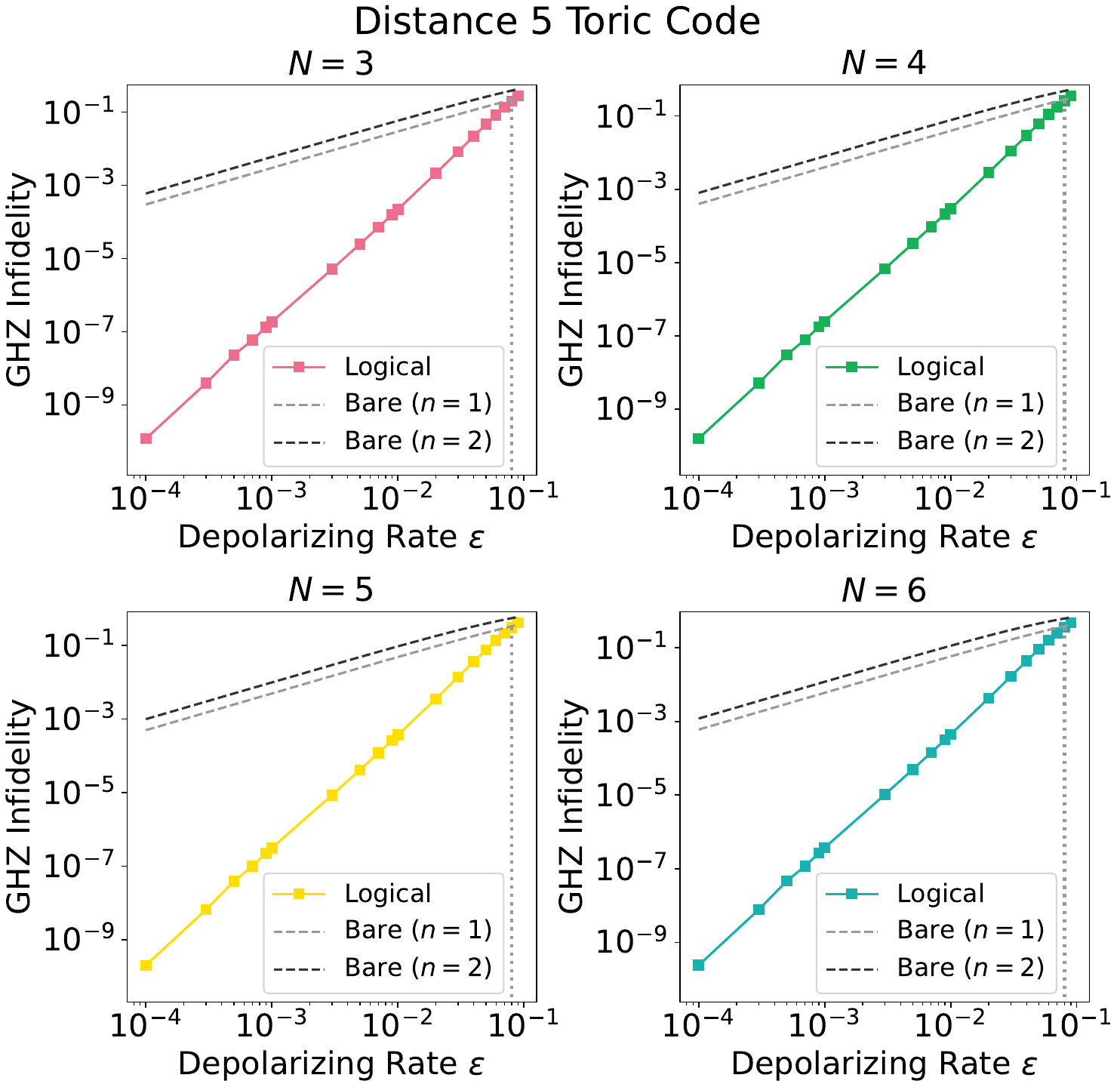}
        \caption{}
        \label{fig:tc_results}
    \end{subfigure}
    \hspace{0.02\textwidth}
    \begin{subfigure}{0.38\textwidth}
        \centering
        \includegraphics[width=\linewidth]{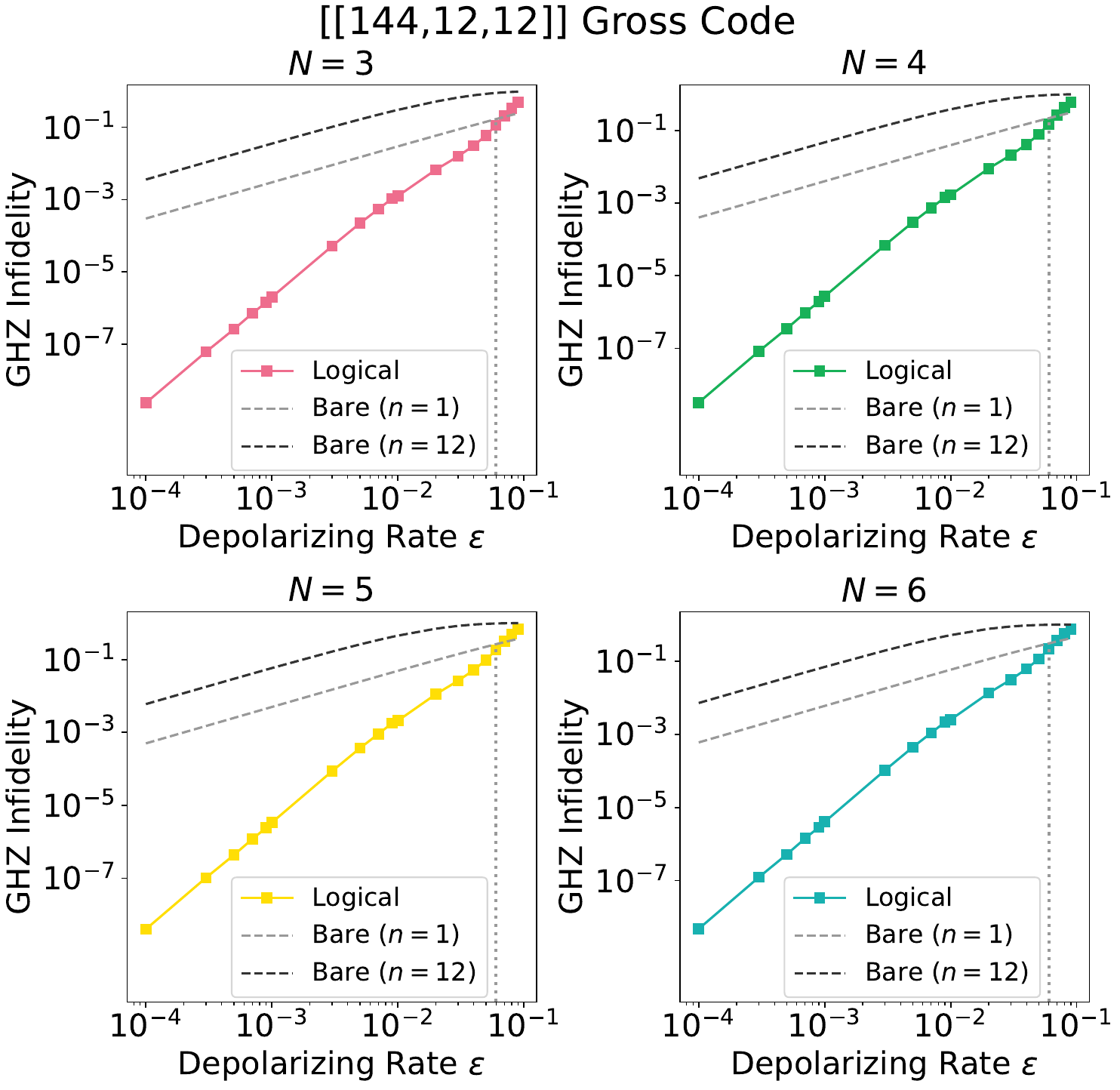}
        \caption{}
        \label{fig:bb_results}
    \end{subfigure}

    \caption{Results for GHZ distillation using (a) the distance 5 toric code and (b) the [[144, 12, 12]] Gross code. We compare to the infidelity of a single bare $N$-qubit GHZ state sent over a depolarizing channel with the same rate (light grey dashed line), as well as the infidelity of sending $k$ copies of the $N$-qubit GHZ state through the noisy channel (dark grey dashed line).}
    \label{fig:big_code_results}
\end{figure}

To assess the performance of our scheme in the ideal (perfect LOCC) case, we simulated the decoding problem for different stabilizer codes, using the depolarizing channel with $p_X = p_Y = p_Z = \varepsilon / 3$ on all physical qubits as a worst-case noise channel. In this section, we assume that the syndrome measurement circuit and recovery operation are realized with no errors. For simplicity in all the simulations presented in this work, each party's noise channel is identical. To benchmark practical realizations of this scheme in a heterogeneous architecture, the channel probabilities can be customized according to experimental details. We leave this to future work.

We calculate the output infidelity $1 - \tilde{f_N}$ by estimating the single-party logical error probability distribution $p(a,b)$ and using \cref{eqn:fidelity_equal} to calculate $\tilde{f_N}$for general $N$. This avoids the computational costs of simulating the full stabilizer tableau of $nN$ qubits and performing $N$ parallel decoding steps for every trial. For small codes, like the 5-qubit code, it is not too expensive to enumerate all of the possible physical errors and their probabilities. We use this approach to arrive at an exact expression for the logical error distribution $p(a,b)$ that would be obtained with a maximum-likelihood decoder. These exact results for the 5-qubit code are shown in \cref{fig:5qubit_results}, where we look both at specific values of $N$ over a range of depolarizing rates, and out to large $N$ with chosen values of $\varepsilon$. In the top row of \cref{fig:5qubit_results}, we mark with a vertical dotted line the point at which the scheme begins to improve the fidelity of the logical GHZ state, compared to the fidelity of sending a single raw GHZ state through a depolarizing channel with the same $\varepsilon$. The grey dashed line marks where the logical infidelity is equal to the bare infidelity. In the bottom row of \cref{fig:5qubit_results}, we show that the scheme improves fidelities by a similar factor over a wide range of $N$ values, when $\varepsilon$ is small. The $\varepsilon=0.2$ plot gives an example of an error rate large enough that the scheme actually worsens performance. With the 5-qubit code, the crossover point where the scheme improves output GHZ state fidelities for all $N$ from 3 to 100 occurs at $\varepsilon \approx 0.138$. For experimentally achievable depolarizing rates on the order of $1\%$, our scheme using the 5-qubit code suppresses the noise by an order of magnitude.

To benchmark performance with larger codes, Monte Carlo simulations with suboptimal decoding algorithms offer an  efficient alternative to exactly calculating $p(a,b)$. We estimate the single-party logical error probability distribution by checking, in each trial, which logical coset the combined noise and correction operation belongs to. The toric code results in \cref{fig:tc_results} were obtained by using minimum-weight perfect matching (MWPM) \cite{dennis2001_Topological,higgott2025_Sparse}. We deployed the BP-OSD decoder to generate the Gross code results show in \cref{fig:bb_results} \cite{roffe2020_Decoding,roffe2022_LDPC}. For each code, we performed between $5 \cdot 10^{8}$ and $5 \cdot 10^{10}$ shots at each depolarizing rate.

\section{Beyond GHZ}
\label{sec:beyond}

Our scheme works for distilling GHZ states precisely because, after the initial round of code generator measurements on each party's qubits, the non-local generators in $H_{tot}'$ are the logical versions of the GHZ stabilizers. We establish in this section which other stabilizer states can be distilled by our scheme in the same way. Specifically, consider taking an arbitrary $N$-qubit stabilizer state given by $\mathcal{S}$ and arranging $n$ copies of the state so that the $A_i$ subsystem contains the $i$th qubit of each state. The collective stabilizer group of this arrangement is $\mathcal{S}_{tot} = \mathcal{S}^{\copymajorsum n}$, which is visualized in \cref{fig:stab_state_ncopies}. Given an $[[n,k,d]]$ stabilizer code $\mathcal{C}$, we ask: when does our same protocol of code generator measurements, noisy qubit distribution, syndrome extraction, and QEC by each party result in $k$ copies of the logical stabilizer state? A formal definition is given in \cref{def:success}. In an initial exploration of this question, we performed an exhaustive numerical search of all $N=3$-qubit stabilizer states, which showed that the 5-qubit code succeeds in distilling all of them. However, other codes failed to distill some stabilizer states. Our \cref{thm:css_states,thm:all_states} establish general compatibility conditions between a state $\mathcal{S}$ and a code $\mathcal{C}$ which explain these observations.

The overall stabilizer group after the initial code generator measurements, $\mathcal{S}_{tot}'$, is generated by $H_{tot}' = H_{\mathcal{C}}^{\oplus N} \cup H_{\mathcal{L}}$. The group of logical operators specified by the rows of $H_{\mathcal{L}}$, which we call $\mathcal{L}$, can be characterized exactly: its generators are $kN$ independent operators from the initial stabilizer group $\mathcal{S}_{tot} = \mathcal{S}^{\copymajorsum n}$ that commute with all of the code generators. Writing the centralizer of $\mathcal{C}^{\oplus N}$ in the $nN$-qubit Pauli group as $C (\mathcal{C}^{\oplus N})$, we have

\begin{equation}
    \label{eqn:logical_stabs}
    \mathcal{L} = \mathcal{S}^{\copymajorsum n} \cap C (\mathcal{C}^{\oplus N}) .
\end{equation}

\begin{definition}
\label{def:success}
    The $[[n,k,d]]$ code $\mathcal{C}$ \textbf{successfully distills} $n$ copies of the stabilizer state $\mathcal{S}$ into $k$ logical copies if and only if the logical subgroup of the post-measurement stabilizer group is the logical version of $k$ copies of $\mathcal{S}$, i.e. if and only if

    \begin{equation}
        \label{eqn:success_defn}
        \pi_{\mathcal{C}}^{\oplus N} ( \mathcal{L} ) = \mathcal{S}^{\copymajorsum k} .
    \end{equation}
\end{definition}

The theorems we prove below ignore the sign column of the parity check matrices, since the sign bits $\alpha^{(i)}$ obtained from the initial code generator measurements are random for each run. The signs of the non-local logical stabilizers in $H_{\mathcal{L}}$ are unaffected by the random measurement outcomes. In proving the Theorems, we construct logical operators for the code $\mathcal{C}$. By fixing these as the logical representatives of each generating coset of the logical group $N(\mathcal{C}) / \mathcal{C}$, independent of which code space of $\mathcal{C}$ the qubits are initially projected into, our results hold for all $\alpha^{(i)}$. The sign bits only affect products of logical representatives with code generators. We now repeat for convenience the form of the code parity check matrix $H_{\mathcal{C}}$ and the state parity check matrix $H_{\mathcal{S}}$, omitting the sign column.

\begin{equation}
    \label{eqn:H_C_no_sign}
    H_{\mathcal{C}} = \left( \begin{array}{c|c}
        X & Z
    \end{array} \right)
\end{equation}

\begin{equation}
    \label{eqn:H_S_no_sign}
    H_{\mathcal{S}} = \left( \begin{array}{c|c}
        A & B
    \end{array} \right)
\end{equation}

\begin{lemma}
    \label{lem:pure_x_z}
    A code $\mathcal{C}$ admits a logical basis of purely X(Z)-type logical X(Z) operators if and only if there exist matrices $U, V \in \mathbb{F}_2^{k \times n}$ such that $XV^T = 0$, $ZU^T = 0$, and $UV^T = I_k$.
\end{lemma}

\begin{proof}
    We prove the sufficient direction first. Suppose the matrices $U, V$ exist as specified, and define $u_j$ ($v_j$) as the $j$th row of $U$ ($V$). Then we can construct $n$-qubit Pauli operators

    \begin{equation}
    \label{eqn:pure_x_z}
        \overline{X_j} = E(u_j, 0) \quad \quad \text{and} \quad \quad \overline{Z_j} = E(0, v_j) \quad j = 1,2, \hdots, k
    \end{equation}

    We now show that the $\overline{X_j}$ and $\overline{Z_j}$ form a valid logical basis for the code $\mathcal{C}$. The condition $ZU^T = 0$ ensures that the $\overline{X_j}$ operators commute with all of the code generators, while $XV^T = 0$ does the same for $\overline{Z_j}$. Finally, $UV^T = I_k$ together with \cref{lem:E_ab_properties}(c) implies the commutation relation

    \begin{equation}
        \label{eqn:logical_commutation}
        \begin{split}
        \overline{X_i} \overline{Z_j} &= E(u_i, 0) E(0, v_j) \\
        &= (-1)^{u_i v_j^T} E(0, v_j) E(u_i, 0) \\
        &= (-1)^{\delta_{ij}} \overline{Z_j} \overline{X_i} .
        \end{split}
    \end{equation}

    For the necessary direction, we assume that a valid basis of purely X/Z logical operators exists. The logical operators therefore have the form in \cref{eqn:pure_x_z}, and we can construct the matrices $U$ and $V$ from the bitstrings $u_j$ and $v_j$.

    \begin{equation}
        \label{eqn:construct_UV_from_rows}
        U = \left( \begin{array}{c}
            u_1 \\
            u_2 \\
            \vdots \\
            u_k
        \end{array} \right) \quad \quad \text{and} \quad \quad V = \left( \begin{array}{c}
            v_1 \\
            v_2 \\
            \vdots \\
            v_k
        \end{array} \right)
    \end{equation}
    
    The conditions on $U$ and $V$ then follow from the fact that these logical operators must commute with all the generators of $\mathcal{C}$ and satisfy the commutation relation $\overline{X_i} \overline{Z_j} = (-1)^{\delta_{ij}} \overline{Z_j} \overline{X_i}$ with each other.
    
\end{proof}

\begin{theorem}
    \label{thm:css_states}
    If a code $\mathcal{C}$ admits a logical basis of purely X(Z)-type logical X(Z) operators, then \cref{def:success} is satisfied for any CSS stabilizer state $\mathcal{S}_{CSS}$.
\end{theorem}

\begin{proof}

    Let a code $\mathcal{C}$ satisfying \cref{lem:pure_x_z} be given. Then the logical operators have the form in \cref{eqn:pure_x_z}. We want to show that $\pi_{\mathcal{C}}^{\oplus N} (\mathcal{L}) = \mathcal{S}_{CSS}^{\copymajorsum k}$, where $\mathcal{L}$ is defined in \cref{eqn:logical_stabs}.

    A CSS stabilizer state has a basis of generators that are all either purely products of Pauli X, or Pauli Z. That is, a CSS stabilizer state parity check matrix can be row reduced to the form

    \begin{equation}
        \label{eqn:H_S_CSS_no_sign}
        H_{\mathcal{S}_{CSS}} = \left( \begin{array}{c|c}
            A' & 0 \\
            0 & B'
        \end{array} \right) \quad \quad A' \in \mathbb{F}_2^{N_X \times N} \quad B' \in \mathbb{F}_2^{N_Z \times N} \quad N_X + N_Z = N
    \end{equation}

    We can write the generators of $\mathcal{S}_{CSS}^{\copymajorsum k}$ explicitly, using $a_{\ell}'$ ($b_{\ell}'$) to denote the $\ell$th row of $A'$ ($B'$).

    \begin{equation}
        \mathcal{S}_{CSS}^{\copymajorsum k} = \langle \{ s_{x}^{(\ell, j)} : \ell = 1, \hdots, N_X; \ \ j = 1, \hdots, k  \} \cup \{ s_{z}^{(\ell, j)} : \ell = 1, \hdots, N_Z; \ \ j = 1, \hdots, k \} \rangle
    \end{equation}

    where
    
    \begin{equation}
        \label{eqn:copies_stab_gen_explicit_css}
        s_x^{(\ell, j)} = \bigotimes_{i=1}^N (X_j)_{A_i}^{(a_{\ell}')_i} \quad \quad \text{and} \quad \quad s_z^{(\ell, j)} = \bigotimes_{i=1}^N (Z_j)_{A_i}^{(b_{\ell}')_i} .
    \end{equation}

    We now define operators $g_x^{(\ell, j)}$ that belong to $\mathcal{L}$, and show that they map to the $s_x^{(\ell, j)}$ under $\pi_{\mathcal{C}}^{\oplus N}$. Recall from the definition of $\mathcal{L}$ that $g \in \mathcal{L}$ if and only if $g \in \mathcal{S}_{CSS}^{\copymajorsum n}$ and $g \in C(\mathcal{C}^{\oplus N})$.

    \begin{equation}
        \label{eqn:g_x_operator}
        g_x^{(\ell, j)} = \bigotimes_{i=1}^N E(u_j, 0)_{A_i}^{(a_{\ell}')_i} \quad \quad \ell = 1, \hdots, N_X; \ \ j = 1, \hdots, k
    \end{equation}

    The logical operator $E(u_j, 0)_{A_i}$ commutes with all code generators on the $A_i$ subsystem, so $g_x^{(\ell, j)} \in C (\mathcal{C}^{\oplus N})$. To show that $g_x^{(\ell, j)} \in \mathcal{S}_{CSS}^{\copymajorsum n}$ as well, we express it as a product of generators of $\mathcal{S}_{CSS}^{\copymajorsum n}$.

    \begin{equation*}
        \begin{split}
            g_x^{(\ell, j)} &= \bigotimes_{i=1}^N \left( \bigotimes_{m=1}^n (X_m)^{(u_j)_m}  \right)_{A_i}^{(a_{\ell}')_i} \\
            &= \bigotimes_{m=1}^n \left( \bigotimes_{i=1}^N (X_m)_{A_i}^{(a_{\ell}')_i}  \right)^{(u_j)_m} \\
            &= \bigotimes_{m=1}^n \left( s_x^{\ell, m} \right)^{(u_j)_m} \in \mathcal{S}_{CSS}^{\copymajorsum n}
        \end{split}
    \end{equation*}

    Now, we apply the logical mapping $\pi_{\mathcal{C}}^{\oplus N}$ to this member of $\mathcal{L}$.

    \begin{equation*}
        \begin{split}
            \pi_{\mathcal{C}}^{\oplus N} \left( g_x^{(\ell, j)} \right) &= \bigotimes_{i=1}^N \left( \pi_{\mathcal{C}} \left( E(u_j, 0) \right) \right)_{A_i}^{(a_{\ell}')_i} \\
            &= \bigotimes_{i=1}^N \left( X_j \right)_{A_i}^{(a_{\ell}')_i} \\
            &= s_x^{(\ell, j)}
        \end{split}
    \end{equation*}

    If we also define

    \begin{equation}
        \label{eqn:g_z_operator}
        g_z^{(\ell, j)} = \bigotimes_{i=1}^N E(0, v_j)_{A_i}^{(b_{\ell}')_i}, \quad \quad \ell = 1, \hdots, N_Z \ \ j = 1, \hdots, k
    \end{equation}

    showing that $\pi_{\mathcal{C}}^{\oplus N} \left( g_x^{(\ell, j)} \right) = s_z^{(\ell, j)}$ proceeds analogously.

    There are $kN$ such $g_{x/z}^{(\ell, j)}$ operators in $\mathcal{L}$. By independence of the logical operators of $\mathcal{C}$ and of the state generators specified by $A'$ and $B'$, all of the $g_{x/z}^{(\ell, j)}$ are independent of each other. Therefore, they must be a complete generating set for the group $\mathcal{L}$, since its parity check matrix has exactly $kN$ rows. By noticing that the $g_{x/z}^{(\ell, j)}$ map to the $s_{x/z}^{(\ell, j)}$ generators of $\mathcal{S}_{CSS}^{\oplus k}$ under $\pi_{\mathcal{C}}^{\oplus N}$, we complete the proof.

\end{proof}

The necessary condition of \cref{thm:css_states} is satisfied by nearly all codes; most notably, every CSS code admits such a logical basis of pure X and Z operators \cite{wilde2009_Logical}. Hence the steps of our protocol serve equally well to distilling $k$ copies of any CSS stabilizer state from $n$ noisy copies of that same state. Besides the $N$-qubit GHZ state, our protocol can therefore produce distributed, high-fidelity logical $\ket{0}^{\otimes k}$ states of any CSS stabilizer code, for example. Toric code ground states could be directly distilled with our protocol. For these beyond-GHZ applications of our protocol, a procedure analogous to the one given in \cref{sec:fidelity} can be used to calculate the final logical state fidelity from single-party decoder performance. The only difference is the form of the logical stabilizer group $\mathcal{S}^{\copymajorsum k}$ containing the allowed logical errors. The lower bound on the success probability given in \cref{eqn:fid_lower_bound} applies equally to all stabilizer states.

We now look beyond CSS stabilizer states and establish a stronger sufficient condition on the logical operators of a code which ensures that \cref{def:success} is satisfied for \textit{all} stabilizer states. Our observation that the 5-qubit code successfully distills all $N=3$-qubit stabilizer states is explained by the $k=1$ special case of \cref{thm:all_states}.

\begin{theorem}
    \label{thm:all_states}
    If a code $\mathcal{C}$ admits a logical basis of purely X(Z)-type logical X(Z) operators with identical support, then \cref{def:success} is satisfied for any stabilizer state $\mathcal{S}$.
\end{theorem}

\begin{proof}
    Let a code $\mathcal{C}$ be given that admits a logical basis of purely X(Z)-type logical X(Z) operators with identical support. This is the special case of \cref{lem:pure_x_z} where $U=V$, i.e. the logical operators of $\mathcal{C}$ have the form

    \begin{equation}
    \label{eqn:pure_x_z_identica}
        \overline{X_j} = E(w_j, 0) \quad \quad \text{and} \quad \quad \overline{Z_j} = E(0, w_j) \quad j = 1,2, \hdots, k.
    \end{equation}

    We make no assumptions about $\mathcal{S}$, so

    \begin{equation}
        \mathcal{S}^{\copymajorsum k} = \langle \{ s^{(\ell, j)} : \ell = 1, \hdots, N; \ \ j = 1, \hdots, k  \} \rangle
    \end{equation}

    where
    
    \begin{equation}
        \label{eqn:copies_stab_gen_explicit}
        s^{(\ell, j)} = i^{a_{\ell} b_{\ell}^T} \bigotimes_{i=1}^N \left( X_j^{(a_{\ell})_i} Z_j^{(b_{\ell})_i} \right)_{A_i} .
    \end{equation}

    Following a similar procedure as in the proof of \cref{thm:css_states}, we define operators $g^{(\ell, j)}$. We show first that they belong to $\mathcal{L}$, and then that they map to the generators of $\mathcal{S}^{\copymajorsum k}$ under $\pi_{\mathcal{C}}^{\oplus N}$.

    \begin{equation}
        g^{(\ell, j)} = i^{a_{\ell} b_{\ell}^T} \bigotimes_{i=1}^N \left( E(w_j, 0)^{(a_{\ell})_i} E(0, w_j)^{(b_{\ell})_i} \right)_{A_i}
    \end{equation}

    We show now that either $g^{(\ell, j)}$ or $ - g^{(\ell, j)}$ is a member of $\mathcal{S}^{\copymajorsum n}$. Being a product of logical operators of $\mathcal{C}$ on each subsystem (up to a global phase), both $\pm g^{(\ell, j)} \in C (\mathcal{C}^{\oplus N})$. In the following, we use $|w_j|$ to denote the Hamming weight of $w_j$, mod 4.

    \begin{equation*}
        \begin{split}
            g^{(\ell, j)} &= i^{a_{\ell} b_{\ell}^T} \bigotimes_{i=1}^N \left( \bigotimes_{m=1}^n \left( X_m^{(w_j)_m} \right)^{(a_{\ell})_i} \left( Z_m^{(w_j)_m} \right)^{(b_{\ell})_i}  \right)_{A_i} \\
            &= i^{a_{\ell} b_{\ell}^T - |w_j| a_{\ell} b_{\ell}^T} \bigotimes_{i=1}^N \left( i^{a_{\ell} b_{\ell}^T} \bigotimes_{m=1}^n \left( X_m^{(a_{\ell})_i} Z_m^{(b_{\ell})_i} \right)  \right)^{(w_j)_m}_{A_i} \\
            &= i^{a_{\ell} b_{\ell}^T - |w_j| a_{\ell} b_{\ell}^T} \bigotimes_{m=1}^n \left( s^{(\ell, m)} \right)^{(w_j)_m} \in \mathcal{S}^{\copymajorsum n} \text{  if } |w_j| = 1
        \end{split}
    \end{equation*}

    If $|w_j| = 3$, we have $ - g^{(\ell, j)} \in \mathcal{S}^{\copymajorsum n}$ instead. The weight of $w_j$ must be odd, because $\overline{X_j}$ and $\overline{Z_j}$ anticommute. For simplicity, we will assume $|w_j| = 1$ for all $j$ in the rest of the proof. In practice, if some $|w_j| = 3$, a logical Pauli frame correction can be applied after the code generator measurements to flip the sign of the corresponding $g^{(\ell, j)}$ generators of the collective stabilizer state. This step would be deterministic, and independent of the random outcomes of the code generator measurements.

    Applying the logical map to $g^{(\ell, j)}$, we see

    \begin{equation*}
        \begin{split}
            \pi_{\mathcal{C}}^{\oplus N} (g^{(\ell, j)} ) &= i^{a_{\ell} b_{\ell}^T} \bigotimes_{i=1}^N \pi_{\mathcal{C}} \left( E(w_j, 0)^{(a_{\ell})_i} E(0, w_j)^{(b_{\ell})_i} \right)_{A_i} \\
            &= i^{a_{\ell} b_{\ell}^T} \bigotimes_{i=1}^N \left( X_j^{(a_{\ell})_i} Z_j^{(b_{\ell})_i} \right)_{A_i} \\
            &= s^{(\ell, j)} .
        \end{split}
    \end{equation*}

    Again, since the state generators of $\mathcal{S}^{\copymajorsum n}$ and the logical operators of $\mathcal{C}$ are independent of each other, the $g^{(\ell, j)}$ form a complete set of the $kN$ generators of $\mathcal{L}$. Hence we conclude that $\pi_{\mathcal{C}}^{\oplus N} (\mathcal{L}) = \mathcal{S}^{\copymajorsum k}$.
    
\end{proof}

\section{Conclusion and Future Work}
\label{sec:disc}

In this work, we developed a QEC-based protocol for distillation of $N$-qubit GHZ states on a network consisting of $N$ quantum processors that are each independently connected to a central state factory. Our scheme utilizes an $[[n,k,d]]$ stabilizer code to distill $n$ initial copies of the GHZ state at the central hub into $k$ copies shared by the end nodes. The steps of the protocol consist of standard QEC operations: local code generator measurements, classical communication of measurement outcomes, and decoding by each party according to their observed error syndrome. We give numerical results for a uniform depolarizing channel on all qubits, which show that the scheme improves the fidelity of the final GHZ state up to a depolarizing rate of $13.8\%$ for the 5-qubit code. For depolarizing rates up to $6 \%$, the larger codes prepare $k >1$ copies of a GHZ state with a higher fidelity than a single bare copy. It is of immediate experimental interest to benchmark performance with other noise models. In particular, a heterogeneous modular architecture with processors utilizing different qubit modalities provides inspiration for further numerical studies. It is possible that the particular asymmetries of such a noise model provide a basis for optimization of machine-learning-based QEC decoders that use the local error syndrome to maximize the final collective state infidelity rather than minimizing the single-party logical error rate. It will also be important to explore the robustness of our protocol against Pauli errors in the initial $n$ copies of the GHZ state. Generalizing the measurement-based Bell pair distillation protocol of \cite{shi2025_MeasurementBased} to $N$-qubit GHZ states may be helpful here, since the initial round of syndrome measurements is avoided.

After presenting our numerical results for GHZ states, we analytically proved that the same steps can successfully distill any $N$-qubit CSS stabilizer state, and we proved that a stronger necessary condition on the $[[n,k,d]]$ code implies successful distillation of all stabilizer states. While we provide the first explicit statement and formal proof of this fact for all CSS stabilizer states, previous works on QEC-based Bell pair and GHZ state distillation have relied on this same insight to guarantee the post-measurement logical state \cite{rengaswamy2022_Distilling,rengaswamy2024_Entanglement,wilde2010_Convolutional}. We provide our full proofs to establish that Bell pairs and GHZ states are a special case of the guarantee for all CSS stabilizer states, and to emphasize the importance of the existence of a purely X/Z logical basis for the $[[n,k,d]]$ code. The notation that we introduce to describe the parity check matrices of distributed stabilizer states may also be helpful for approaching other quantum networking problems with the stabilizer formalism.

\section*{Acknowledgements}
The authors are grateful to Narayanan Rengaswamy for helpful discussions and feedback on a draft of the manuscript. This material is based upon work supported by the Defense Advanced Research Projects Agency (DARPA) under Agreement No. HR0011-26-9-E114. Approved for public release; distribution is unlimited.

\clearpage

\printbibliography

\clearpage

\crefalias{section}{appendix}

\appendix

\section{Stabilizer State Update Rules}
\label{sec:stab_rules}

Consider an $N$-qubit stabilizer state $\ket{\psi}$ described by the stabilizer group $\mathcal{S} = \langle s_1, s_2, ..., s_N \rangle$. The update rules for unitary evolution and Pauli measurement are as follows:

\begin{enumerate}
    \item Unitary evolution: $\ket{\psi} \mapsto U \ket{\psi}$ is described by $s_j \mapsto U s_j U^{\dagger} \ \ \forall j \in \{1, 2, ..., N\}$. When $U$ is a Pauli operator, the $s_j$ are unchanged if they commute with $U$, and are only changed by a sign $s_j \mapsto - s_j$ if they anticommute. When $U$ belongs to the Clifford group, the $s_j$ remain Pauli operators, which enables the evolution to be simulated efficiently on a classical computer (the Gottesman-Knill theorem).
    \item Measurement of a Pauli operator $V$, obtaining the result $\pm 1$:
    \begin{enumerate}[label=(\alph*), leftmargin=4em]
        \item If $[s_j, V] = 0$ for all $j \in \{1, 2, ..., N\}$, then $V$ or $-V$ is already in $\mathcal{S}$. No update is needed, and the measurement result is deterministic.
        \item If $[s_k, V] \neq 0$ for exactly one $k \in \{1, 2, ..., N\}$, then remove $s_k$ from the set of stabilizer generators and replace it with $\pm V$.

            \begin{equation*}
                s_k \mapsto s_k' = \pm V, \quad s_j \mapsto s_j \ \forall j \neq k
            \end{equation*}
        
        \item If $[s_j, V] \neq 0$ for some subset of the stabilizer generators, i.e. $j \in \mathcal{J} \subseteq \{ 1, 2, ..., N\}$, then choose any one $s_k$ for $k \in \mathcal{J}$ to replace with $\pm V$. Also update the other non-commuting generators as $s_j \mapsto s_j s_k$ for all $j \in \mathcal{J} - \{k \}$ 

            \begin{equation*}
                s_k \mapsto s_k' = \mu V, \quad s_j \mapsto s_j' = s_j s_k \ \forall j \in \mathcal{J} - \{k \}, \quad s_i \mapsto s_i \ \forall i \in \{ 1, 2, ..., N\} - \mathcal{J}
            \end{equation*}
    \end{enumerate}
\end{enumerate} 

\section{The GHZ Map}
\label{sec:ghz_map}

When there are $n$ copies of the GHZ state, and each party has one qubit from each copy, the state vector has the form

\begin{equation}
    \label{eqn:ncopies_ghz_state}
    \ket{\text{GHZ}^{(N)}_n} = \ket{\text{GHZ}^{(N)}}^{\otimes n} = \frac{1}{\sqrt{2^n}} \sum_{x \in \mathbb{F}_2^n} \ket{x}_{A_1} \ket{x}_{A_2} ... \ket{x}_{A_N}.
\end{equation}

GHZ states are multipartite entangled states. For instance, it can be immediately read from \cref{eqn:ncopies_ghz_state} that a Z-basis measurement of the $A_1$ subsystem will project all subsystems into the same basis state, specified by the measurement outcome $x \in \mathbb{F}_2^n$. More generally, Rengaswamy et al. derive a ``GHZ Map" between any operation on the $n$ qubits of one subsystem (initially in the state $\ket{\text{GHZ}^{(N)}_n}$) and an equivalent operation acting jointly on the qubits of the other $N-1$ subsystems \cite[Lemma 4]{rengaswamy2024_Entanglement}. For our purposes in particular, the relevant operation on one subsystem is the projection realized by the measurement of a code generator. The equivalent operation is a projection of the qubits in the other subsystems whose form can be calculated directly, starting from \cref{eqn:ncopies_ghz_state} \cite{rengaswamy2024_Entanglement}.

We now state and discuss the GHZ Map results proven in \cite{rengaswamy2024_Entanglement} before showing that the same result can be obtained by applying the update rules of the stabilizer formalism.

\begin{theorem}
\label{thm:reng_N3}
    \cite[Theorem 6]{rengaswamy2024_Entanglement} Given $n$ copies of the 3-qubit GHZ state shared between Alice, Bob, and Charlie, measuring $E(a,b)_A = E([a,0,0], [b,0,0])$ on Alice's $n$ qubits and obtaining the result $\varepsilon \in \{ \pm 1 \}$ is equivalent to measuring the following with results $+1$ on the qubits of Bob and Charlie:

    \begin{equation*}
        \begin{split}
            \varepsilon (-1)^{ab^T} E(a,b)_B & \otimes E(a,0)_C  = \varepsilon (-1)^{ab^T} E([0,a,a], [0,b,0]) \text{ and } \\
            &\{ Z_{Bj} Z_{Cj} : j = 1,2,...,n \}.
        \end{split}
    \end{equation*}
\end{theorem}

The exact meaning of the equivalence mentioned in the statement of \cref{thm:reng_N3} is made explicit by the equality in \cref{eqn:ghz_map_N3_explicit} below.

\begin{equation}
    \label{eqn:ghz_map_N3_explicit}
    \begin{split}
    \frac{I_{3n} + \varepsilon E([a,0,0],[b,0,0])}{2} \ket{\text{GHZ}^{(3)}_n} = \ &\frac{I_{3n} + \varepsilon (-1)^{ab^T} E([0,a,a], [0,b,0])}{2} \ \cdot  \\ & \quad \quad \left( \prod_{j=1}^n \frac{I_{3n} + Z_{Bj} Z_{Cj}}{2} \right) \ket{\text{GHZ}^{(3)}_n}
    \end{split}
\end{equation}

We know from the stabilizer group of the GHZ state that $\ket{\text{GHZ}^{(3)}_n}$ is a $+1$ eigenstate of all the $Z_{Bj} Z_{Cj}$ operators. Hence the product over those projectors acts as as the identity on the GHZ state and can be ignored. Furthermore, using the fact that projectors square to themselves, one can show that \cref{eqn:ghz_map_N3_functional} follows from \cref{eqn:ghz_map_N3_explicit}.

\begin{equation}
    \label{eqn:ghz_map_N3_functional}
    \left( \frac{I_n + \varepsilon E(a,b)}{2} \right)_A \ket{\text{GHZ}^{(3)}_n} = \left( \frac{I_n + \varepsilon E(a,b)}{2} \right)_A \otimes \left( \frac{I_{2n} + \varepsilon (-1)^{ab^T} E([a,a],[b,0])}{2} \right)_{BC} \ket{\text{GHZ}^{(3)}_n}
\end{equation}

\cref{eqn:ghz_map_N3_functional} is the clearest statement of the functional purpose of \cref{thm:reng_N3}. Measurement of $E(a,b)$ on the $A$ qubits simultaneously projects the $B$ and $C$ qubits into an eigenspace of a $2n$-qubit operator $E([a,a],[b,0])$. The claim is generalized to arbitrary $N$ by \cref{thm:reng_N}. Of course, since the GHZ state is permutation invariant, the choice of $A$ (or $A_1$) as the measured subsystem is made purely for notational clarity.

\begin{theorem}
\label{thm:reng_N}
    \cite[Theorem 7]{rengaswamy2024_Entanglement} Given $n$ copies of the $N$-qubit GHZ state with subsystems $A_1, A_2, ..., A_N$, measuring $E(a,b)_{A_1}$ on the $n$ qubits of subystem $A_1$ and obtaining the result $\varepsilon \in \{ \pm 1 \}$ is equivalent to measuring the following with results $+1$ on the qubits of the remaining $(N-1)$ subystems:
    
    \begin{equation*}
        \begin{split}
            \varepsilon (-1) & ^{\left( b + \sum_{i=1}^{N-2} \sum_{j>i}^{N-1} b_i * b_j \right) a^T} \bigotimes_{t=2}^N E(a,b_t)_{A_t} = \varepsilon (-1)^{\left( b + \sum_{i=1}^{N-2} \sum_{j>i}^{N-1} b_i * b_j \right) a^T} E([0,a,a,...,a], [0, b_2, b_3, ..., b_N]) \\
            &\quad \quad \quad \quad \quad \text{and }\{ Z_{(A_2)j} Z_{(A_3)j}, Z_{(A_3)j} Z_{(A_4)j}, ..., Z_{(A_{N-1})j} Z_{(A_N)j} : j = 1,2,...,n \},
        \end{split}
    \end{equation*}

    where $b_2, b_3, ..., b_N \in \mathbb{F}_2^n$ satisfy $b_2 \oplus b_3 \oplus ... \oplus b_n = b$, $\oplus$ denotes addition modulo 2, and $x * y$ denotes the element-wise product of 2 vectors.
    
\end{theorem}

\begin{remark}
    To simplify the expression, one can choose $b_2 = b$ and $b_3 = ... = b_N = 0$, so that $b_i * b_j = 0$ always, making the induced stabilizer $\varepsilon (-1)^{ab^T} E([0,a,a,...,a], [0,b,0,0,...,0])$. 
\end{remark}

This characterization of the simultaneous projection of the unmeasured subsystems when a Pauli operator is measured on one subsystem admits a straightforward interpretation through the lens of the stabilizer formalism. We now introduce a different notation for the generators of $\mathcal{G}_N^{\copymajorsum n}$ than we used in the main text. The full stabilizer group is generated by $n$ independent generating sets - one for each copy indexed by $i$ in \cref{eqn:ghz_stab_Nn}. Here, $e_j \in \mathbb{F}_2^n$ denotes the canonical basis vector with a $1$ in the $j$th position and $0$ elsewhere.

\begin{equation}
\label{eqn:ghz_stab_Nn}
    \mathcal{G}_N^{\copymajorsum n} = \langle E([e_j, e_j, ..., e_j], 0), E(0, [e_j, e_j, 0, ..., 0]), ..., E(0, [0, ..., 0, e_j, e_j]) : j = 1,2,...,N \rangle
\end{equation}

Recall that \cref{thm:reng_N} asserts that when the $n$-qubit operator $E(a,b)_{A_1}$ is measured on the $A_1$ subsystem and the result $\varepsilon = \pm 1$ is obtained, the non-measured parties $A_2, ..., A_N$ are projected into the $\varepsilon (-1)^{ab^T}$ eigenspace of the operator $E([0,a,a,...,a], [0,b,0,...,0])$. We now establish with \cref{thm:stab} that this induced projection can be viewed a natural consequence of the stabilizer formalism. The full expression in \cref{thm:reng_N}, including the arbitrary choice of $b_2, b_3, ..., b_N$ and the additional sign factors involving $b_i * b_j$, can be obtained by additionally taking the product with any combination of the $Z_{(A_t)i} Z_{(A_{t+1})i}$ operators $(t \geq 2)$, which are also generators for the post-measurement stabilizer group.

\begin{theorem}
    \label{thm:stab}
    After the local code stabilizer $E(a,b)_{A_1}$ is measured on the $A_1$ subsystem ($a,b \in \mathbb{F}_2^n$), obtaining the result $\varepsilon = \pm 1$, the post-measurement state is the stabilizer state for the updated stabilizer group $\mathcal{G}'$. There is an operator in the updated stabilizer group with joint support on all unmeasured subsystems that has the form $\varepsilon (-1)^{ab^T} E([0,a,a,...,a], [0,b,0,...,0])$.
\end{theorem}

\begin{proof}

Initially, we have $n$ copies of an $N$-qubit GHZ state shared across subsystems $A_1, A_2, ..., A_N$. The collective state is the stabilizer state of the group $\mathcal{G}_N^{\copymajorsum n}$, whose generators are given in \cref{eqn:ghz_stab_Nn}.

The measured operator $E(a,b)_{A_1} = E([a,0,...,0], [b,0,...,0])$ becomes a generator of $\mathcal{G}'$ with the observed sign $\varepsilon$, by the stabilizer update rules discussed in \cite{gottesman1998_Heisenberg}

We now calculate the other generators of the updated stabilizer group $\mathcal{G}'$. Only the generators of $\mathcal{G}_N^{\copymajorsum n}$ which anticommute with $E(a,b)_{A_1}$ are changed; let $W$ denote the set of such generators, which we can characterize exactly.

\begin{equation*}
    W = \{ E(0, [e_j, e_j, 0, ..., 0]) : j \text{ s.t. } a_j = 1  \} \cup \{ E([e_j, e_j, ..., e_j], 0) : j \text{ s.t. } b_j = 1 \} 
\end{equation*}

To be explicit, note that we can express the generating set of $\mathcal{G}_N^{\copymajorsum n}$ in terms of $W$ in the following way.

\begin{equation*}
    \begin{split}
    \mathcal{G}_N^{\copymajorsum n} = \langle W \cup & \{ E(0, [e_j, e_j, 0, ..., 0]) : j \text{ s.t. } a_j = 0  \} \cup \{ E([e_j, e_j, ..., e_j], 0) : j \text{ s.t. } b_j = 0 \} \cup \\
    & \{ E(0, [0, e_j, e_j, 0, ..., 0]), E(0, [0, 0, e_j, e_j, 0, ..., 0]), ..., E(0, [0, ..., 0, e_j, e_j]) : j = 1, 2, ..., N  \} \rangle
    \end{split}
\end{equation*}

In this discussion we will assume $a, b \neq 0$ for maximum generality. The cases with $a=0$ or $b=0$ are straightforward to approach with the same strategy, and the case $a=b=0$ is trivial because $E(0,0)$ is the identity. Since $W$ is the set of generators of $\mathcal{G}_N^{\copymajorsum n}$ which anticommute with the measured operator, the stabilizer update rules \cite{gottesman1998_Heisenberg} give us

\begin{equation*}
    \begin{split}
    \mathcal{G}' = \langle W' \cup & \{ E(0, [e_j, e_j, 0, ..., 0]) : j \text{ s.t. } a_j = 0  \} \cup \{ E([e_j, e_j, ..., e_j], 0) : j \text{ s.t. } b_j = 0 \} \cup \\
    & \{ E(0, [0, e_j, e_j, 0, ..., 0]), E(0, [0, 0, e_j, e_j, 0, ..., 0]), ..., E(0, [0, ..., 0, e_j, e_j]) : j = 1, 2, ..., N  \} \rangle
    \end{split}
\end{equation*}

where

\begin{equation*}
    \begin{split}
        W' = \{ \varepsilon E([a,0,...,0], & [b,0,...,0]) \} \cup \{ E(0, [e_j + e_k, e_j + e_k, 0, ..., 0] : j \neq k \text{ s.t. } a_j = 1 \} \cup \\ &\{ (-1)^{\delta_{jk}} E([e_j, e_j, ..., e_j], [e_k, e_k, 0, ..., 0]) : j \text{ s.t. } b_j = 1 \}.
    \end{split}
\end{equation*}

Here, $k$ is one fixed index such that $a_k = 1$, which always exists since we assume $a \neq 0$, and $\delta_{jk}$ is the Kroenecker delta. Now, with the generators of the post measurement stabilizer group $\mathcal{G}'$ explicitly defined, we can go about expressing the desired operator, $\varepsilon (-1)^{ab^T} E([0, a, ..., a], [0, b, 0, ..., 0])$, as a product of these generators. We will refer to the generators of the form $E(0, [e_j, e_j, 0, ..., 0])$ and $E(0, [e_j + e_k, e_j + e_k, 0, ..., 0])$ as ``Z-type" generators, and we let ``X-type" generators refer to $E([e_j, e_j, ..., e_j], 0)$ and $(-1)^{\delta_{jk}} E([e_j, e_j, ..., e_j], [e_k, e_k, 0, ..., 0])$.

Consider first the operator $B$, the product of all the Z-type stabilizers corresponding to a value of $j$ in the support of $b$. By closure of the stabilizer group, $B \in \mathcal{G}'$.

\begin{equation*}
\begin{split}
    B &= \left( \prod_{\substack{j \text{ s.t.} \\ a_j = 0, b_j = 1}} E(0, [e_j, e_j, 0, ..., 0]) \right) \left( \prod_{\substack{j \neq k \text{ s.t. } \\ a_j = b_j = 1}} E(0, [e_j+e_k, e_j + e_k, 0, ..., 0]) \right) \\
    &= E(0, [c, c, 0, ..., 0]) \left( \prod_{\substack{j \neq k \text{ s.t.} \\ a_j = b_j = 1}} E(0, [e_j+e_k, e_j + e_k, 0, ..., 0]) \right) \quad \quad \quad c = \bigoplus_{\substack{j \text{ s.t.} \\ a_j=0 \\ b_j=1}} e_j
\end{split}
\end{equation*}

We now introduce a shorthand for taking an integer modulo 2, $\widetilde{x} := x \text{ mod } 2$, and we note the following identity which will be helpful in evaluating the second product.

\begin{equation*}
    \widetilde{|\{ j : j \neq k \text{ and } a_j = b_j = 1 \}|}
        =
        \left\{
        \begin{array}{ll}
            \widetilde{ab^T} & \text{if } b_k = 0 \\
            \widetilde{\left(ab^T + 1\right)} & \text{if } b_k = 1
        \end{array}
        \right\}
        =
        \widetilde{\left(ab^T + b_k\right)}
\end{equation*}

The second product evaluates to

\begin{equation*}
    \prod_{\substack{j \neq k \text{ s.t.} \\ a_j = b_j = 1}} E(0, [e_j+e_k, e_j + e_k, 0, ..., 0]) = E(0, [d \oplus \widetilde{\left(  ab^T + b_k \right)} e_k, d \oplus \widehat{\left(  ab^T + b_k \right)}e_k, 0, ..., 0)])  \quad \quad \quad d = \bigoplus_{\substack{j\neq k \text{ s.t.} \\ a_j=b_j=1}} e_j
\end{equation*}

Note that by the definitions of $c$ and $d$, their sum is

\begin{equation*}
    c \oplus d = \bigoplus_{\substack{j \neq k \text{ s.t.} \\ b_j = 1}} e_j =  \left\{
        \begin{array}{ll}
            b & \text{if } b_k = 0 \\
            b \oplus e_k & \text{if } b_k = 1
        \end{array}
        \right\}
        = b \oplus b_k e_k .
\end{equation*}
    
This allows us to arrive at a closed form for $B$, the product of all Z type generators with $j$ in the support of $b$.

\begin{align*}
    B &= E(0, [c,c,0,...,0]) \cdot E(0, [d \oplus \widetilde{\left( ab^T + b_k \right)} e_k, d \oplus \widetilde{\left( ab^T + b_k \right)} e_k, 0, ..., 0]) \\
    &= E(0, [c \oplus d \oplus \widetilde{\left( ab^T + b_k \right)} e_k, c \oplus d \oplus \widetilde{\left( ab^T + b_k \right)} e_k, 0, ..., 0]) \\
    &= E(0, [b \oplus b_k e_k \oplus \widetilde{\left( ab^T + b_k \right)} e_k, b \oplus b_k e_k \oplus \widetilde{\left( ab^T + b_k \right)} e_k, 0, ..., 0]) \\
    &= E(0, [b \oplus \widetilde{\left( ab^T \right)} e_k, b \oplus \widetilde{\left( ab^T \right)} e_k, 0, ..., 0])
\end{align*}

We now calculate $A$, the product of all the X-type generators corresponding to $j$ in the support of $a$. Again, of course, $A \in \mathcal{G}'$ by closure.

\begin{equation*}
\begin{split}
    A &= \left( \prod_{\substack{j \text{ s.t.} \\ a_j = 1, b_j = 0}} E([e_j, e_j, ..., e_j], 0) \right) \left( \prod_{\substack{j \text{ s.t. } \\ a_j = b_j = 1}} (-1)^{\delta_{jk}}E([e_j, e_j, ..., e_j], [e_k, e_k, 0, ..., 0]) \right) \\
    &= E([f, f, ..., f], 0) \left( \prod_{\substack{j \text{ s.t. } \\ a_j = b_j = 1}} (-1)^{\delta_{jk}}E([e_j, e_j, ..., e_j], [e_k, e_k, 0, ..., 0]) \right)  \quad \quad \quad f = \bigoplus_{\substack{j \text{ s.t.} \\ a_j=1 \\ b_j=0}} e_j \\
    &= E([f, f, ..., f], 0) \left( \prod_{\substack{j \text{ s.t. } \\ a_j = b_j = 1}} (-1)^{\delta_{jk}} \right)  \left( \prod_{\substack{j \text{ s.t. } \\ a_j = b_j = 1}} E([e_j, e_j, ..., e_j], [e_k, e_k, 0, ..., 0]) \right) \\
\end{split}
\end{equation*}

First we can reason out what the product over $(-1)^{\delta_{jk}}$ must be. Recall that $k$ was defined as some index for which $a_k = 1$. Whether or not a term with $\delta_{jk} = 1$ is included in the product is determined by whether or not $b_k = 1$, since the product runs over $j$ such that $a_j =1 \text{ and } b_j = 1$. Any terms with $j \neq k$ contribute nothing to the product, since $(-1)^0 = 1$. If $b_k = 1$, then we get one $(-1)^1$ term, giving an overall minus sign. In other words,

\begin{equation*}
    \prod_{\substack{j \text{ s.t. } \\ a_j = b_j = 1}} (-1)^{\delta_{jk}} =  \left\{
        \begin{array}{ll}
            +1 & \text{if } b_k = 0 \\
            -1 & \text{if } b_k = 1
        \end{array}
        \right\}
        = (-1)^{b_k}.
\end{equation*}

Now we consider the product of Pauli operators. Defining $g$ as the bitstring with a 1 at all $j$ such that $a_j = b_j = 1$, and 0 elsewhere, we can see that $f \oplus g = a$ always. This is contrasted with the bistrings we defined while calculating $B$, which only obeyed $c \oplus d = b$ if $b_k = 0$. It can be shown following the rules for multiplication of Pauli operators defined in \cref{sec:bkgd_pauli}, that

\begin{equation*}
    \prod_{\substack{j \text{ s.t. } \\ a_j = b_j = 1}} E([e_j, e_j, ..., e_j], [e_k, e_k, 0, ..., 0]) = \sigma E([g, g, ..., g], [\widetilde{\left( ab^T \right)} e_k, \widetilde{\left( ab^T \right)} e_k, 0, ..., 0]),
\end{equation*}

where

\begin{equation*}
    \sigma =  \left\{
        \begin{array}{ll}
            \prod_{\substack{j \text{ s.t. } \\ a_j = b_j = 1}} (-1)^{\delta_{jk}} & \text{if } \widetilde{ab^T} = 0 \\
            +1 & \text{if } \widetilde{ab^T} = 1
        \end{array}
        \right\}
        = \left\{
        \begin{array}{ll}
            (-1)^{b_k} & \text{if } \widetilde{ab^T} = 0 \\
            +1 & \text{if } \widetilde{ab^T} = 1
        \end{array}
        \right. .
\end{equation*}

The overall sign in front of $E([g, g, ..., g], [\widetilde{\left( ab^T \right)} e_k, \widetilde{\left( ab^T \right)} e_k, 0, ..., 0])$ is therefore

\begin{equation*}
    \nu := (-1)^{b_k} \sigma = \left\{
        \begin{array}{ll}
            +1 & \text{if } \widetilde{ab^T} = 0 \\
            (-1)^{b_k} & \text{if } \widetilde{ab^T} = 1
        \end{array}
        \right\} = \left\{
        \begin{array}{ll}
            -1 & \text{if } \widetilde{ab^T} = 1 \text{ and } b_k = 1 \\
            +1 & \text{else}
        \end{array}
        \right.
\end{equation*}

We can now evaluate the closed form of $A$ as a Pauli operator in the $E$ notation.

\begin{equation*}
\begin{split}
    A &= \nu E([f, f, ..., f], 0) \cdot E([g, g, ..., g], [\widetilde{\left( ab^T \right)} e_k, \widetilde{\left( ab^T \right)} e_k, 0, ..., 0]) \\
    &= \nu \mu E([a,a,...,a], [\widetilde{\left( ab^T \right)} e_k, \widetilde{\left( ab^T \right)} e_k, 0, ..., 0])
\end{split}
\end{equation*}

Where the sign $\mu$ depends on whether $f_k = 1$.

\begin{equation*}
    \mu = (-1)^{(\widetilde{ab^T}) \cdot f_k} =  \left\{
        \begin{array}{ll}
            -1 & \text{if } \widetilde{ab^T} = 1 \text{ and } b_k = 0 \\
            +1 & \text{else}
        \end{array}
        \right.
\end{equation*}

If $\widetilde{ab^T} = 1$, exactly one of $\mu$ and $\nu$ will be $-1$, controlled by whether $b_k$ is 0 or 1. And if $\widetilde{ab^T} = 0$, then $\mu = \nu = +1$. Hence their overall product gives us the $(-1)^{ab^T}$ that we expect from the GHZ map.

\begin{equation*}
    A = (-1)^{ab^T} E([a,a,...,a], [\widetilde{\left( ab^T \right)} e_k, \widetilde{\left( ab^T \right)} e_k, 0, ..., 0])
\end{equation*}

We now calculate the product $AB$ before finally multiplying by the measured operator to arrive at the claim of \cref{thm:stab}.

\begin{align*}
    AB &= (-1)^{ab^T} E([a,a,...,a], [\widetilde{\left( ab^T \right)} e_k, \widetilde{\left( ab^T \right)} e_k, 0, ..., 0]) \cdot E(0, [b \oplus \widetilde{\left( ab^T \right)} e_k, b \oplus \widetilde{\left( ab^T \right)} e_k, 0, ..., 0]) \\
    &= (-1)^{ab^T} \lambda E([a,a,...,a], [b,b,0,...,0])
\end{align*}

Where the sign $\lambda$ emerging from the multiplication rule is always $+1$. We show this below, using the crucial fact that $a_k = 1$ by definition of $k$.

\begin{equation*}
    \lambda = (-1)^{a \left( b \oplus \widetilde{\left( ab^T \right)} e_k \right)^T} = (-1)^{\widetilde{\left( ab^T \right)} + \widetilde{\left( ab^T \right)} a e_k^T} = (-1)^{\widetilde{\left( ab^T \right)} + \widetilde{\left( ab^T \right)} a_k} = (-1)^{2 \widetilde{\left( ab^T \right)}} = +1
\end{equation*}

So, by calculating a particular product of generators for the post-measurement stabilizer group, we have shown that $(-1)^{ab^T} E([a,a,...,a], [b,b,0,...,0]) \in \mathcal{G}'$. By multiplying this operator by one last generator of $\mathcal{G}'$, namely the measured operator (with the obtained sign), $\varepsilon E([a, 0, ..., 0], [b,0,...,0])$, we complete the proof.

\begin{align*}
    \varepsilon E([a,0,...,0],[b,0,...,0]) \cdot AB &= \varepsilon (-1)^{ab^T} E([a,0,...,0], [b,0,...,0]) \cdot E([a,a,...,a], [b,b,0,....,0]) \\
    &= \varepsilon (-1)^{ab^T} \cdot i^{ab^T - ab^T} \cdot E([2a,a,a,...,a], [2b, b, 0, 0, ..., 0]) \\
    &= \varepsilon (-1)^{ab^T} E([0, a, a, ..., a], [0, b, 0, 0, ..., 0]) \in \mathcal{G}'
\end{align*}

\end{proof}

\subsection{Examples}

For concreteness in the following examples, we give in \cref{eqn:ghz_stab_33} the full generating set for $n=3$ copies of a GHZ state shared by Alice, Bob, and Charlie.

\begin{equation}
    \label{eqn:ghz_stab_33}
    \begin{split}
        \mathcal{G}_{3}^{\copymajorsum 3} = \langle \ & Z_{A1} Z_{B1}, Z_{B1} Z_{C1}, X_{A1} X_{B1} X_{C1}, \\
        &Z_{A2} Z_{B2}, Z_{B2} Z_{C2}, X_{A2} X_{B2} X_{C2}, \\
        &Z_{A3} Z_{B3}, Z_{B3} Z_{C3}, X_{A3} X_{B3} X_{C3} \rangle
    \end{split}
\end{equation}

\begin{example}

Suppose Alice measures the local operator $E(a,b)_A = E([1,0,0], [0,1,0])_A = X_{A1} Z_{A2}$ and obtains the result $\varepsilon = \pm 1$. We expect, then, that Bob and Charlie's qubits should be projected into the $\varepsilon$ eigenspace of $E([a,a], [b,0]) = (X_{B1} Z_{B2}) (X_{C1})$.

\end{example}

The measured operator anticommutes with $Z_{A1} Z_{B1}$ and $X_{A2} X_{B2} X_{C2}$, while it commutes with all of the other generators. Applying the update rules, then, we have

\begin{equation*}
    \begin{split}
        &Z_{A1} Z_{B1} \mapsto \varepsilon X_{A1} Z_{A2} \\
        &X_{A2} X_{B2} X_{C2} \mapsto X_{A2} X_{B2} X_{C2} \cdot (Z_{A1} Z_{B1}) = (Z_{A1} X_{A2}) (Z_{B1} X_{B2}) (X_{C2}).
    \end{split}
\end{equation*}

So the full set of generators for the updated stabilizer group $\mathcal{G}'$ is

\begin{equation*}
\begin{split}
    \mathcal{G}' =  \langle \ & \varepsilon X_{A1} Z_{A2}, (Z_{A1} X_{A2}) (Z_{B1} X_{B2}) (X_{C2}),  \\
        &Z_{B1} Z_{C1}, X_{A1} X_{B1} X_{C1}, Z_{A2} Z_{B2}, Z_{B2} Z_{C2},\\
        &Z_{A3} Z_{B3}, Z_{B3} Z_{C3}, X_{A3} X_{B3} X_{C3} \rangle
\end{split}
\end{equation*}

The operator corresponding to the induced projection, $\varepsilon E([1,0,0,1,0,0], [0,1,0,0,0,0])_{BC} = \varepsilon (X_{B1} Z_{B2}) (X_{C1})$, is readily expressed as a product of these generators, which shows it is in the stabilizer group for the state after $X_{A1} Z_{A2}$ is measured. This in turn implies that the post-measurement state is in the $+1$ eigenspace of $\varepsilon (X_{B1} Z_{B2}) (X_{C1})$, which is the desired result.

\begin{equation*}
        (\varepsilon X_{A1} Z_{A2}) \cdot (X_{A1} X_{B1} X_{C1}) \cdot (Z_{A2} Z_{B2}) = \varepsilon (X_{B1} Z_{B2}) (X_{C1}) \in \mathcal{G}'
\end{equation*}

\begin{example}

In this example we demonstrate that the stabilizer interpretation is consistent with the sign factor $(-1)^{ab^T}$. Consider again the $N=n=3$ GHZ state specified by $\mathcal{G}_3^{\copymajorsum 3}$. This time, Alice measures $E(a,b)_A = E([1,0,0], [1,0,0])_A = Y_{A1}$ and obtains the outcome $\varepsilon$. Now, since $ab^T = 1$, we expect Bob and Charlie's qubits to lie in the $ - \varepsilon$ eigenspace of $E([a,a], [b,0])_{BC} = Y_{B1} X_{C1}$.

\end{example}

The original generators $Z_{A1} Z_{B1}$ and $X_{A1} X_{B1} X_{C1}$ anticommute with $Y_{A1}$.

\begin{equation*}
    \begin{split}
        &Z_{A1} Z_{B1} \mapsto \varepsilon Y_{A1} \\
        &X_{A1} X_{B1} X_{C1} \mapsto X_{A1} X_{B1} X_{C1} \cdot (Z_{A1} Z_{B1}) = (iY_{A1}) (iY_{B1}) (X_{C1}) = - (Y_{A1}) (Y_{B1}) (X_{C1}).
    \end{split}
\end{equation*}

So the post-measurement stabilizer group is

\begin{equation*}
    \begin{split}
        \mathcal{G}' = \langle \ & \varepsilon Y_{A1}, - Y_{A1} Y_{B1} X_{C1}, Z_{B1} Z_{C1}, \\
        &Z_{A2} Z_{B2}, Z_{B2} Z_{C2}, X_{A2} X_{B2} X_{C2}, \\
        &Z_{A3} Z_{B3}, Z_{B3} Z_{C3}, X_{A3} X_{B3} X_{C3} \rangle.
    \end{split}
\end{equation*}

Now we show that $ - \varepsilon Y_{B1} X_{C1} \in \mathcal{G}'$ by expressing it as a product of the updated generators, which proves that the post-measurement state lies in the $ - \varepsilon$ eigenspace of $Y_{B1} X_{C1}$.

\begin{equation*}
    (\varepsilon Y_{A1}) \cdot (- Y_{A1} Y_{B1} X_{C1}) = - \varepsilon Y_{B1} X_{C1}  \in \mathcal{G}'
\end{equation*}

\section{Success Probability from Single-Party QEC}
\label{sec:fidelity_details}

Recall the following expressions from \cref{sec:fidelity} involving the overall logical operation $L$ applied by the correction step.

\begin{equation*}
    L = \bigotimes_{i=1}^N E(a_i, b_i)_{A_i} = E([a_1, a_2, ..., a_N], [b_1, b_2, ..., b_N]) \quad \quad a,b \in \mathbb{F}_2^k
\end{equation*}

\begin{equation*}
    P(L) = \prod_{i=1}^N p_i (a_i, b_i).
\end{equation*}

\begin{equation*}
    p_{\text{succ}}^{(N)} = \sum_{L \in \mathcal{G}_N^{\copymajorsum k}} P(L) .
\end{equation*}

From the form of the state generators for the GHZ state, it can be checked that $L \in \mathcal{G}_N^{\copymajorsum k}$ if and only if $a_1 = a_2 = ... = a_N$ and $b_1 \oplus b_2 \oplus ... \oplus b_N = 0$. Therefore

\begin{equation}
\label{eqn:p_succ_direct}
    p_{\text{succ}}^{(N)} = \sum_{a} \sum_{\substack{b_1, ..., b_N s.t. \\ b_1 \oplus ... \oplus b_N = 0}} \prod_{i=1}^N p_i (a, b_i)
\end{equation}

To make further progress, we introduce a variable that evaluates to $1$ whenever $b_1 \oplus b_2 \oplus ... \oplus b_N = 0$, and $0$ otherwise, which we call $\mathbf{1}[b_1 \oplus b_2 \oplus ... \oplus b_N = 0]$.

\begin{equation*}
    \begin{split}
        \mathbf{1}[b_1 \oplus b_2 \oplus ... \oplus b_N = 0] &= \frac{1}{2^k} \sum_{s \in \mathbb{F}_2^k} (-1)^{s (b_1 \oplus b_2 \oplus ... \oplus b_N)^T} \\
        &= \frac{1}{2^k} \sum_{s \in \mathbb{F}_2^k} \prod_{i=1}^N (-1)^{s b_i^T}
    \end{split}
\end{equation*}

So 

\begin{equation*}
\begin{split}
    p_{\text{succ}}^{(N)} &= \sum_{a} \sum_{b_1, ..., b_N}  \mathbf{1}[b_1 \oplus b_2 \oplus ... \oplus b_N = 0] \prod_{i=1}^N p_i (a, b_i) \\
    &= \frac{1}{2^k} \sum_{a} \sum_{b_1, ..., b_N} \sum_s \left( \prod_{i=1}^N (-1)^{s b_i^T} \right) \left( \prod_{i=1}^N p_i (a, b_i)  \right) \\
    &=  \frac{1}{2^k} \sum_{a} \sum_s \sum_{b_1, ..., b_N} \left( \prod_{i=1}^N (-1)^{s b_i^T} p_i (a, b_i) \right) \\
    &= \frac{1}{2^k} \sum_a \sum_s \prod_{i=1}^N \left( \sum_b (-1)^{s b^T} p_i (a,b) \right)
\end{split}
\end{equation*}

To condense notation slightly, we define the quantity

\begin{equation*}
    q_i (a, s) := \sum_b (-1)^{s b^T} p_i (a, b),
\end{equation*}

so that

\begin{equation*}
    p_{\text{succ}}^{(N)} = \frac{1}{2^k} \sum_a \sum_s \prod_{i=1}^N q_i (a, s) .
\end{equation*}

When the $p_i$ are identical across all parties $A_i$, then this form reduces to \cref{eqn:fidelity_equal}. When there are $M \leq N$ distinct logical error probability distributions, we must separately calculate all $M$ corresponding $q_i$. However, this just requires an $\mathcal{O}(N)$ overhead in the worst case, whereas the cost of naively iterating through all of the $b_i$ in \cref{eqn:p_succ_direct} is exponential in $N$. The exponential cost in $k$ is unavoidable as the number of possible logical errors grows as $4^k$. However, at low physical error rates, high weight logical errors are suppressed, and we save overhead by only iterating through those $E(a,b)$ logical errors that are observed in the simulations.

\end{document}